\documentclass{article}
\usepackage{amsmath}
\usepackage{amsfonts}
\usepackage{amssymb}
\usepackage{amsthm}
\usepackage[numbers]{natbib}
\usepackage[utf8]{inputenc}
\usepackage{graphicx}
\usepackage{url}
\usepackage{array,multirow,makecell}
\usepackage[T1]{fontenc}
\usepackage{setspace}
\usepackage{accents}
\usepackage{color}
\usepackage{enumitem}
\usepackage[11pt]{extsizes}
\usepackage{layout}
\usepackage[top=2cm, bottom=2cm, left=2cm, right=2cm]{geometry}
\usepackage[colorlinks=true,linkcolor=blue]{hyperref}
\usepackage{systeme}
\usepackage{pgfplots}
\pgfplotsset{width=8cm,compat=1.9}
\usepackage{comment}
\usepackage{tabularray}
\usepackage{pgfplotstable}
\usepackage{xspace}
\usepgfplotslibrary{statistics}
\usepackage{algpseudocode}
\usepackage{algcompatible}
\usepackage[ruled, lined, linesnumbered, commentsnumbered, longend]{algorithm2e}
\usepackage{xcolor}
\usepackage{authblk}
\usepackage{bm}
\usepackage{appendix}
\usepackage{fancyvrb}
\usepackage[pagewise]{lineno}
\usepackage{subcaption}
\usepackage{tikz}
\usepackage{adjustbox}

\SetCommentSty{mycommfont}

\newcommand{\q}[1]{``#1''}
\DeclareMathOperator*{\esssup}{ess\,sup}
\DeclareMathOperator{\diag}{diag}

\DeclareGraphicsExtensions{.pdf,.png}
\renewcommand{\q}[1]{``#1''}

\theoremstyle{plain}
\newtheorem{theorem}{Theorem}[section]
\newtheorem{Proposition}[theorem]{Proposition}

\newtheorem{remark}[theorem]{Remark}

\title{Swing contract pricing: with and without Neural Networks}

\author[1]{Vincent Lemaire}
\author[1]{Gilles Pagès}
\author[1,2]{Christian Yeo}
\affil[1]{\footnotesize Sorbonne Université, Laboratoire de Probabilités, Statistique et Modélisation, UMR 8001, case 158, 4, pl. Jussieu,
F-75252 Paris Cedex 5, France}
\affil[2]{\footnotesize Engie Global Markets, 1 place Samuel Champlain, 92400 Courbevoie, France}

\date{}
\numberwithin{equation}{section}

\begin{document}

%\linenumbers
\maketitle

\begin{abstract}
    We propose two parametric approaches to evaluate swing contracts with firm constraints. Our objective is to define approximations for the optimal control, which represents the amounts of energy purchased throughout the contract. The first approach involves approximating the optimal control by means of an explicit parametric function, where the parameters are determined using stochastic gradient descent based algorithms. The second approach builds on the first one, where we replace parameters in the first approach by the output of a neural network. Our numerical experiments demonstrate that by using Langevin based algorithms, both parameterizations provide, in a short computation time, better prices compared to state-of-the-art methods.
\end{abstract}


\section*{Introduction}
\indent

With the energy market becoming increasingly deregulated, various derivative products have emerged, offering flexibility in delivery dates and purchased amounts of energy. \textit{Swing} contracts, also known as \textit{Take-or-Pay} contracts (see \cite{Thompson1995ValuationOP} for more details), are among the most widely traded contracts in gas and power markets. These contracts allow their holder to purchase amounts of energy at fixed exercise dates, subject to constraints. There exist two types of constraints: firm or penalized. In the firm constraints setting, the contract holder is not allowed to violate the constraints. In the penalized setting, violating the (global) constraints is penalized proportionally to the deficit/excess of consumption. Our paper focuses on the valuation of swing contracts with firm constraints.

The valuation of such swing contracts is more challenging than that of classic American-style contracts \cite{Broadie2004ASM, Jaillet1990VariationalIA, Longstaff2001ValuingAO, Rogers2002MonteCV}, precisely because of the presence of constraints related to time and volume. From a probabilistic standpoint, the valuation of swing contracts appears as a Stochastic Optimal Control (\textit{SOC}) problem. Here, the control represents the vector of volumes of energy to purchase at each exercise date. To solve this SOC problem, two groups of methods may be considered.

The first group concerns methods based on the so-called \textit{Backward Dynamic Programming Principle} (BDPP). In this group, the price of the swing contract can be written as the solution of a dynamic programming equation (see \cite{Bardou2007WhenAS, Bardou2009OptimalQF, BarreraEsteve2006NumericalMF,Jaillet2004ValuationOC, LariLavassani2002ADV,Thompson1995ValuationOP}). BDPP approach involves a conditional expectation also known as the \q{continuation value} and its numerical computation is the main difficulty. The most commonly used method for computing the continuation value is the Longstaff and Schwartz one (see \cite{Longstaff2001ValuingAO}) where it is approximated as successive orthogonal projection over subspaces spanned by a finite number of squared integrable random variables (see \cite{BarreraEsteve2006NumericalMF, Longstaff2001ValuingAO}). Another method is based on optimal quantization (see \cite{Bardou2009OptimalQF, Pags2004OptimalQM}) where the stochastic dynamics of the traded asset is approximated by space discretization of the whole process taking the shape of the so-called quantization tree. This tree is built offline by nearest neighbour projections of simulated paths of the underlying dynamics. One of the major problems encountered with BDPP approaches is the following. To compute the value of the contract at each exercise date, the maximisation is performed on geometric interval. In practice, this interval needs to be discretized leading to a loss of accuracy. Indeed, achieving a high level of accuracy in pricing the contract requires a finer discretization of the geometric interval, which in turn increases the computation time. Additionally, the Longstaff-Schwartz method faces a storage challenge since, to compute the continuation value, regression coefficients for each simulation and admissible cumulative consumption must be stored at each exercise date. This often leads to a memory overflow when a large number of simulations is required to obtain a price with a tight confidence interval. Meanwhile, the optimal quantization method suffers from the \q{curse of dimensionality} because it is well known that the rate of convergence of this method is of order $\mathcal{O}(N^{-\frac{1}{d}})$ (where $N$ is the number of points of the quantization tree and $d$ is the problem dimension).

An alternative approach to the BDPP one is to consider the valuation of swing contracts as a global optimization problem. As already mentioned, the valuation of swing contracts is equivalent to solve a \textit{SOC} problem where the objective is to find a vector of purchase amounts that maximizes the expected value of cumulative cash flows up to the expiry of the contract. These cash flows have to be simulated leading to a stochastic optimization problem where the set of admissible controls is assimilated to a parameterized family of functions, the optimization being performed by a Stochastic Gradient Descent based on (Monte Carlo) simulations of the underlying dynamics. The use of parametric functions to approximate the optimal control reduces the control problem to a parametric optimization problem. The success of such approaches is highly depending on the choice of the parameterized family of functions which often requires a deep understanding of the optimal control behaviour. Note that, solving \textit{SOC} problem with SGD based algorithms has already been considered in \cite{Bachouch_2021, PierreLangevin, GobetMunos, Han2016DeepLA, 9446979} for general \textit{SOC} problems and more specially in \cite{BarreraEsteve2006NumericalMF} for swing contracts. In this paper, we introduce two parameterized families of functions to approximate the optimal control in the swing valuation problem and where one of the parameterizations is based on deep neural networks. To optimize both parameterizations, we consider two optimization algorithms based on stochastic gradient descent. The first optimization algorithm is the classic \q{Adaptive Moment Estimation} (Adam) algorithm which is widely used for stochastic approximation. As an alternative to Adam algorithm, we propose the use of the Preconditioned Stochastic Gradient Langevin Dynamics (PSGLD) which has demonstrated his effectiveness in recent studies for Bayesian learning. Our results show that using Langevin-type algorithms can accelerate the optimization of our parameterizations and leads to better prices compared to the state-of-the-art methods.

Our paper is organised as follows. Section \ref{sec1}. We describe swing contracts and recall the pricing framework. Section \ref{sec2}. We present our contribution by proposing two parameterized families of functions designed to approximate the optimal control. Section \ref{training_part}. We present generalities about stochastic approximation as well as two optimization algorithms (Adam and PSGLD) to optimize both parameterized controls. Finally in section \ref{sec5}, we present some additional results about the effectiveness of our methods.

\section{On swing contracts}
\label{sec1}

\indent
Due to their dual constraint (exercise dates and volume constraints) nature, swing contracts present some specificities when it comes to their evaluation. This section provides a reminder of the general background on swing contracts.

\subsection{Description}

\indent
A swing option is a commonly encountered derivative product on energy (gas and power) markets which allows its holder to buy amounts of energy $q_{\ell}$ at times $t_\ell = \frac{\ell T}{n}$, $\ell = 0, \ldots,n-1$ (called exercise dates) until the contract maturity $t_n = T$. At each exercise date $t_\ell$, the unitary purchase price (or strike price of the contract) is denoted $K_\ell$. It can be constant (i.e. $K_\ell = K, \ell = 0,\ldots,n-1$) or indexed on the price of either the underlying instrument or another one (oil for example). In this paper, we focus on the fixed strike price case but the case with indexed strike can be treated likewise.

One important feature of a swing option are the volume constraints. Swing option gives its holder a flexibility on the amount of energy he is allowed to purchase and this amount is subject to two kinds of (firm) constraints
\begin{itemize}
    \item \textbf{Local constraints}: at each exercise date $t_\ell$, the holder of the swing contract is allowed to buy at least $q_{\min}$ and at most $q_{\max}$ of the underlying i.e.,
    \begin{equation}
    \label{loc_const}
    0 \le q_{\min} \le q_{\ell} \le q_{\max}, \hspace{0.5cm} 0 \le \ell \le n-1.
    \end{equation}

    \item \textbf{Global constraints}: the cumulative consumption up to the maturity of the contract must be not lower than $Q_{\min}$ and not greater than $Q_{\max}$ i.e.,
    \begin{equation}
    \label{glob_const}
    Q_{n} = \sum_{\ell = 0}^{n-1} q_{\ell} \in [Q_{\min}, Q_{\max}], \hspace{0.4cm} \text{with} \hspace{0.2cm} Q_0 = 0 \hspace{0.2cm} \text{and} \hspace{0.2cm} 0 \le Q_{\min} \le Q_{\max} \le +\infty.
    \end{equation}
\end{itemize}

In this paper, we only consider \textbf{firm constraints} which means that the holder of the contract is not allowed to violate the constraints. In this setting, the existence of an optimal (\textit{bang-bang}) consumption had been established in \cite{Bardou2007WhenAS}. A similar result is proved for the case with penalties in \cite{BarreraEsteve2006NumericalMF}. Recall that the penalty setting corresponds to the case where the holder of the swing contract is allowed to violate the global constraints but is faced to penalties proportional to the excess/deficit of consumption (purchased volumes of energy).

\vspace{0.2cm}
In the firm constraints case, the local and global constraints allow to draw the swing physical space which represents possible actions per exercise date.

\subsection{Physical space}
\label{physical_space}

\indent
At each exercise date, due to firm constraints, the attainable cumulative volumes are upper and lower bounded by two functions $t_\ell \mapsto Q^{down}(t_\ell)$ and $t_\ell \mapsto Q^{up}(t_\ell)$

\vspace{0.1cm}
\begin{equation}
\left\{
    \begin{array}{ll}
        Q^{down}(t_0) = 0,\\
        \displaystyle Q^{down}(t_\ell) = \max\big(0, Q_{\min} - (n-\ell) \cdot q_{\max} \big) \text{,} \hspace{0.2cm} \ell \in \{1,\ldots,n-1\}, \\
        Q^{down}(t_n) = Q_{\min},
    \end{array}
\right.
\end{equation}

\vspace{0.1cm}
\noindent
and

\begin{equation}
\left\{
    \begin{array}{ll}
        Q^{up}(t_0) = 0,\\
        \displaystyle Q^{up}(t_\ell) = \min\big(\ell \cdot q_{\max}, Q_{\max}\big)\text{,} \hspace{0.2cm} \ell \in \{1,\ldots,n-1\},\\
        Q^{up}(t_n) = Q_{\max},
    \end{array}
\right.
\end{equation}

\vspace{0.2cm}

\noindent
These boundaries lead to the physical space of the swing contract, as drawn in Figure \ref{physical space swing}: they represent at each exercise date, the range of attainable cumulative consumptions.
\newpage

\begin{figure}[!ht]
    \center
    \includegraphics[scale=0.5]{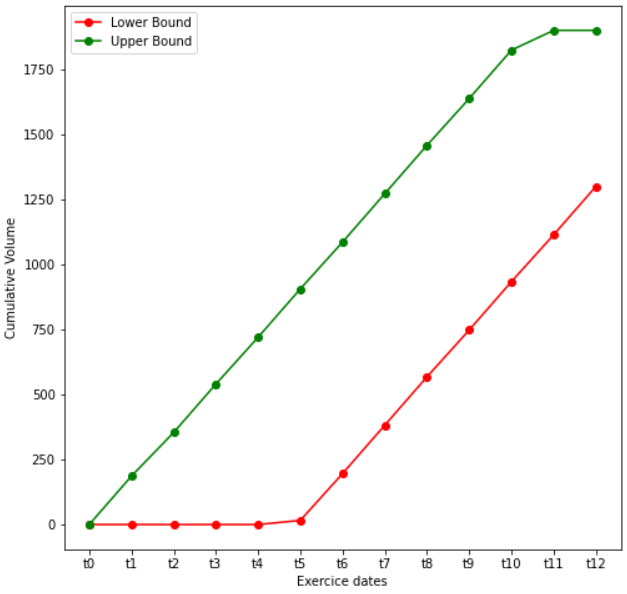}
    \caption{Illustration of the swing physical space with $q_{\min} = 0, q_{\max} = 180, Q_{\min} = 1300, Q_{\max} = 1900$ and $n=12$ exercise dates.}
    \label{physical space swing}
\end{figure}

\noindent
In Figure \ref{physical space swing}, we assumed that $q_{\min} = 0$. This assumption can be made without loss of generality because, as shown in Appendix \ref{swing decompo} (proof reproduced for the reader convenience), the general case can always reduce to the case where $q_{\min} = 0$ and $q_{\max} = 1$. Besides, if at an exercise date $t_\ell$, we have bought amounts $q_{0},...,q_{\ell-1}$ leading to a cumulative consumption $\displaystyle Q_\ell = \sum_{i = 0}^{\ell -1} q_i$, then due to local constraints and the cumulative consumption boundaries, the actual range of attainable cumulative consumption at time $t_{\ell+1}$ is given by
\begin{equation}
    \Big[\underbrace{\max(Q^{down}(t_{\ell+1}), Q_{\ell} + q_{\min})}_{:= L_{{\ell+1}}(Q_\ell)}, \hspace{0.2cm}  \underbrace{\min(Q^{up}({\ell+1}), Q_{\ell} + q_{\max})}_{:= U_{{\ell+1}}(Q_\ell)}\Big].
    \label{range_cum_vol}
\end{equation}

\noindent
Thus, starting with a cumulative consumption $Q_{\ell}$, at time $t_\ell$, the range of admissible volumes (to be purchased) is
\begin{equation}
    \Big[\underbrace{L_{{\ell+1}}(Q_\ell) - Q_\ell}_{:= A_\ell^{-}(Q_\ell)}, \hspace{0.2cm}  \underbrace{U_{{\ell+1}}(Q_\ell) - Q_\ell}_{:= A_\ell^{+}(Q_\ell)}\Big].
    \label{intervalle_control}
\end{equation}

It is important to note the \textit{bang-bang} feature of the swing contracts in the firm constraints setting. It had been established in \cite{Bardou2007WhenAS} that when volume constraints are integers and $Q_{\max} - Q_{\min}$ is a multiple of $q_{\max} - q_{\min}$, then the optimal consumption (volume to purchase) is bang-bang. The latter means that, at each exercise date $t_\ell$, the optimal consumption $q_\ell$ lies within $\{A_\ell^{-}, A_\ell^{+} \}$ (see \eqref{intervalle_control}). The bang-bang feature is particularly valuable since it significantly reduces the computation time.

\vspace{0.2cm}

With these basic blocks, we can now set up the theoretical framework for the pricing of swing contracts.

\subsection{Pricing and sensitivity calculus}
\label{pricing}

\indent
The price at time $t$ of the forward contract delivered at the maturity $T$ is denoted by $F_{t, T}$. It is worth noting that, in the firm constraints setting aforementioned, there is no action to be done at the maturity $T$ whereas in the case with penalty (see \cite{BarreraEsteve2006NumericalMF}), a penalty is applied, at the maturity $T$, in case of violation of global constraints defined in \eqref{glob_const}. This entails that, in practice and in the firm constraints setting, the underlying asset $F_{t, T}$ needs to be diffused until $t_{n-1}$ unlike the case with penalty where we need to diffuse until $t_n = T$. We keep the notation $t_{n-1}$ to denote the last exercise date following previous works \cite{Bardou2007WhenAS, Bardou2009OptimalQF, BarreraEsteve2006NumericalMF}.

In this paper (as in \cite{Bardou2009OptimalQF, BarreraEsteve2006NumericalMF}), we consider contracts on the spot price ($S_t = F_{t,t}$) even if in practice the spot turns out not to be a tradable instrument on market. In practice, the most encountered gas contract is the day-ahead forward contract whose price is  $F_{t, t+1}$. But this case can be treated in the same way. We also consider a probability space $\big(\Omega, \mathcal{F}, \{ \mathcal{F}_{t_\ell}, 0 \le \ell \le n-1 \}, \mathbb{P} \big)$ where $(\mathcal{F}_{t_\ell})_{0\le \ell \le n-1}$ is the natural completed filtration of the process $\big(S_{t_\ell}\big)_{0\le \ell \le n-1}$.

The decision process $(q_{\ell})_{0 \le \ell \le n-1}$ is defined on the same probability space and is supposed $\mathcal{F}_{t_\ell}$- adapted. At each exercise date $t_\ell$, by buying a volume $q_\ell$, the holder of the contract makes an algebraic gross profit (or loss)
\begin{equation}
\psi_\ell\left(q_{\ell}, S_{t_\ell} \right) := q_{\ell} \cdot\left(S_{t_\ell} - K\right).
\end{equation}

Given a cumulative consumption $Q \in \mathbb{R}_{+}$, an admissible strategy at time $t_k$ is a vector $(q_k, \ldots, q_{n-1})$ lying in the following set
\begin{equation*}
\mathcal{A}_{k, Q}^{Q_{\min}, Q_{\max}} = \left\{(q_{\ell})_{k \le \ell \le n-1}, \hspace{0.1cm} q_{\ell} : (\Omega, \mathcal{F}_{t_{\ell}}, \mathbb{P}) \mapsto [q_{\min}, q_{\max}], \hspace{0.1cm} \sum_{\ell = k}^{n-1} q_{\ell} \in \big[(Q_{\min}-Q)_{+}, Q_{\max}-Q \big] \right\}.
\end{equation*}

\noindent
Using \eqref{intervalle_control}, this set also reads,
\begin{equation}
\label{actual_set_adm_const_case}
\mathcal{A}_{k, Q}^{Q_{\min}, Q_{\max}} = \left\{(q_{\ell})_{k \le \ell \le n-1}, \hspace{0.1cm} q_{\ell} : (\Omega, \mathcal{F}_{t_{\ell}}, \mathbb{P}) \mapsto \big[A_\ell^{-}(Q_\ell), A_\ell^{+}(Q_\ell) \big], \hspace{0.1cm} \text{where} \hspace{0.1cm} Q_\ell = Q + \sum_{i = k}^{\ell - 1} q_i \right\},
\end{equation}

\noindent
with the convention $\displaystyle \sum_{i = k}^{k - 1} q_i = 0$. Then for every non negative $\mathcal{F}_{t_{k-1}}$- measurable random variable $Q$, the price of the swing contract at time $t_k$ starting from a cumulative consumption $Q$ is given by
\begin{equation}
    P_k\left(S_{t_k}, Q \right) = \mathbb{P}-\esssup_{(q_{\ell})_{k \le \ell \le n-1} \in \mathcal{A}_{k, Q}^{Q_{\min}, Q_{\max}}} \hspace{0.1cm} \mathbb{E}\left(\sum_{\ell=k}^{n-1} e^{-r_\ell(t_{\ell} - t_k)} \psi_\ell\left(q_{\ell}, S_{t_{\ell}} \right) \rvert \mathcal{F}_{t_k} \right),
    \label{pricing_swing_formula}
\end{equation}

\noindent
where the expectation is taken under the risk-neutral probability $\mathbb{P}$ and $(r_{\ell})_\ell$ are interest rates that we will assume to be zero from now on. The price of the swing contract, at time $t_0 = 0$, is given by ($S_{t_0}$ is assumed to be deterministic)
\begin{equation}
    P_0 := \underset{(q_\ell)_{0 \le \ell \le n-1} \in \mathcal{A}_{0, 0}^{Q_{\min}, Q_{\max}}}{\sup} \hspace{0.1cm}  \mathcal{J}\big(q_0, \ldots, q_{n-1}\big),
    \label{swing_init_price}
\end{equation}

\noindent
where, given an admissible strategy $(q_0, \ldots, q_{n-1})$, the reward function $\mathcal{J}$ is defined as the expected value of cumulative future cash flows up to the expiry
\begin{equation}
\label{obj_func_J}
\mathcal{J}\big(q_0, \ldots, q_{n-1}\big) := \mathbb{E}\left(\sum_{\ell=0}^{n-1} \psi_{\ell}\big(q_{\ell}, S_{t_{\ell}} \big)\right).
\end{equation}

The problem \eqref{swing_init_price} appears to be a constrained stochastic control problem in which the aim is to find an admissible strategy (consumption) that maximizes the reward function $\mathcal{J}$. To achieve this, we propose two parametric strategies. Unless otherwise stated, for numerical computations, we use a one-factor model like in \cite{BarreraEsteve2006NumericalMF, Jaillet2004ValuationOC}:
\begin{equation}
	\frac{dF_{t, T}}{F_{t_, T}} = \sigma_F \cdot e^{-\lambda(T-t)}dW_t, \hspace{0.4cm}  t \le T,
	\label{hjm_model}
\end{equation}

\noindent
where $(W_t)_{t \ge 0}$ is a standard Brownian motion. As mentioned before, we deal with the spot price whose price can be derived by a straightforward application of Itô
\begin{equation}
\label{spot_model}
S_t = F_{0,t} \cdot \exp\big(\sigma_F \cdot X_t - \frac{1}{2}\Sigma_t^2\big), \hspace{0.3cm} X_t = \int_{0}^{t} e^{-\lambda (t-s)}\, \mathrm{d}W_s \hspace{0.3cm} \text{and} \hspace{0.3cm} \Sigma_t^2 = \frac{\sigma^2}{2\lambda}\big(1-e^{-2\lambda t}\big).
\end{equation}

\newpage

\noindent
Along with this model, we consider two settings (presented below). In both cases, we set $\lambda = 4, \sigma_F = 0.7$.

\vspace{0.3cm}
\textbf{Case 1}

\begin{center}
$n = 31$ exercise dates \hspace{0.4cm} $q_{\min} = 0 \hspace{0.4cm} q_{\max} = 6 \hspace{0.4cm} Q_{\min} = 140 \hspace{0.4cm} Q_{\max} = 200$. 
\end{center}

\textbf{Case 2} (as in \cite{Bardou2009OptimalQF, BarreraEsteve2006NumericalMF})

\begin{center}
$n =365$ exercise dates \hspace{0.4cm} $q_{\min} = 0 \hspace{0.4cm} q_{\max} = 6 \hspace{0.4cm} Q_{\min} = 1300 \hspace{0.4cm} Q_{\max} = 1900$.
\end{center}

The computer used for all computations and simulations has the following characteristics: \textit{Processor: Intel(R) Core(TM) i7-1185G7 @ 3.00GHz, 32Go of RAM, Microsoft Windows 10 Enterprise}.

\vspace{0.2cm}

Besides, for practitioners, prices of derivative products are closely linked to their sensitivities with respect to market data. These sensitivities are essential for hedging purposes, as they indicate how the price of a derivative product changes with the market. Computing sensitivities involves calculating derivatives of functions, and in our case, we need to differentiate a price with respect to some parameters of the underlying diffusion, where the price is given by \eqref{swing_init_price}. To achieve this, we rely on the so-called \textit{envelope theorem}. Let us briefly recall a short background of this theorem.

\vspace{0.2cm}

Let $f(x,\alpha)$ and  $g_{j}(x,\alpha),j=1,2,\ldots ,m$ be real-valued continuously differentiable functions on $\mathbb{R}^{n+\ell}$ , where $x\in \mathbb {R} ^{n}$ are some variables and $\alpha \in \mathbb {R}^{\ell}$ are parameters, and consider the constrained optimization problem
$$
\left\{
    \begin{array}{ll}
        \underset{x}{\max} \hspace{0.1cm} f(x, \alpha)\\
        \text{subject to} \hspace{0.2cm} g_{j}(x,\alpha) \ge 0, \hspace{0.4cm} 1 \le j \le m.
    \end{array}
\right.
$$

\noindent
We introduce the Lagrangian function,
$$\mathcal{L}\left(x, \mu, \alpha \right) = f(x, \alpha) + \langle\mu, g(x, \alpha) \rangle,$$

\noindent
where $\mu \in \mathbb{R}^{m}$ is the Lagrange multipliers, $\langle \cdot,\cdot \rangle$ is the Euclidean inner-product and $g = (g_1, \ldots, g_m)^\top$ (for a vector $x$, $x^\top$ denotes its transpose). Then we define the value function $V(\alpha) = f(x^{*}(\alpha), \alpha)$ where $x^{*}(\alpha)$ is a solution that maximizes the function $f(\cdot, \alpha)$. The following theorem gives the derivative of the value function $V$ in case it is differentiable.

\begin{theorem}[Envelope theorem]
\label{env_thm}
Assume that $V$ and $\mathcal {L}$ are continuously differentiable where $\frac{\partial \mathcal{L}}{\partial \alpha_k} = \frac{\partial f}{\partial \alpha_k} + \langle\mu, \frac{\partial g}{\partial \alpha_k}\rangle$. Then,

$$\frac{\partial V(\alpha)}{\partial \alpha_k} = \frac{\partial \mathcal{L}\left(x^{*}(\alpha), \mu^{*}(\alpha), \alpha  \right)}{\partial \alpha_k}, \hspace{0.7cm} k=1,\ldots,\ell.$$
\end{theorem}

In model \eqref{hjm_model}, we deduce the following proposition which is a corollary of Theorem \ref{env_thm}.

\begin{Proposition}
Let $\left(q_{k}^{*} \right)_{0 \le k \le n-1}$ be a solution of the problem \eqref{swing_init_price}. Assume that the function
$$\left(F_{0, t_k} \right)_{0 \le k \le n-1} \mapsto \mathbb{E}\left(\sum_{k = 0}^{n-1} q_{k}^{*} \cdot \left( F_{0, t_k} \cdot e^{\sigma_F \cdot X_{t_k} - \frac{1}{2}\Sigma^2_{t_k}} - K \right) \right)$$

\noindent
is continuously differentiable. Let $P_0$ be the price of the swing contract \eqref{swing_init_price}. Note that $P_0$ is a function of $\big(S_{t_k})_{0 \le k \le n-1}$ which in turns depends on $\big(F_{0,t_k})_{0 \le k \le n-1}$ in the model \eqref{hjm_model}. Then for all $k = 0,\ldots, n-1$, the delta (sensitivity of swing price with respect to the initial forward price) is given by
$$\frac{\partial P_0}{\partial F_{0, t_k}} = \mathbb{E}\left(q_{k}^{*} \cdot e^{\sigma_F \cdot X_{t_k} - \frac{1}{2}\Sigma^2_{t_k} } \right),$$

\noindent
where $X_{t_k}$ is defined in \eqref{spot_model}.
\end{Proposition}

\begin{proof}
We define the following functions:
$$f\left((q_{k})_{0 \le k \le n-1}, (F_{0, t_k})_{0 \le k \le n-1} \right) := \mathbb{E}\left(\sum_{k = 0}^{n-1} q_{k} \cdot \left( F_{0, t_k} \cdot e^{\sigma_F \cdot X_{t_k} - \frac{1}{2}\Sigma^2_{t_k}} - K \right) \right),$$
$$g_1\left((q_{k})_{0 \le k \le n-1}, (F_{0, t_k})_{0 \le k \le n-1}\right) := \sum_{k=0}^{n-1} q_{k} - Q_{\min},$$
$$g_2\left((q_{k})_{0 \le k \le n-1}, (F_{0, t_k})_{0 \le k \le n-1}\right) := Q_{\max}- \sum_{k=0}^{n-1} q_{k} ,$$

\noindent
and for all $k = 0, \ldots, n-1$:
$$g_{2k+3}\left((q_{k})_{0 \le k \le n-1}, (F_{0, t_k})_{0 \le k \le n-1}\right) := q_{k} - q_{\min},$$
$$g_{2k+4}\left((q_{k})_{0 \le k \le n-1}, (F_{0, t_k})_{0 \le k \le n-1}\right) := q_{\max} - q_{k}.$$

\noindent
Then it suffices to prove that functions $f, g_1, g_2, \ldots, g_{2n+2}$ are continuously differentiable in order to use the envelope theorem. This holds for functions $g_1, g_2, \ldots, g_{2n+2}$. It remains to prove it for function $f$. The latter reduces to show that function $f$ is continuously differentiable in each of its components. For any $k = 0, \ldots, n-1$, the random variable $q_{k} \cdot \big(F_{0, t_k} \cdot e^{\sigma_F \cdot X_{t_k} - \frac{1}{2}\Sigma^2_{t_k}} - K \big)$ is integrable since
$$\Big|q_{k} \cdot \big(F_{0, t_k} \cdot e^{\sigma_F \cdot X_{t_k} - \frac{1}{2}\Sigma^2_{t_k}} - K \big)\Big| \le q_{\max} \cdot \big(F_{0, t_k} \cdot e^{\sigma_F \cdot X_{t_k} - \frac{1}{2}\Sigma^2_{t_k}} + K \big) \in \mathbb{L}_\mathbb{R}^1(\mathbb{P}).$$

\noindent
Moreover, the function $F_{0, t_k} \mapsto q_{k} \cdot\big(F_{0, t_k} \cdot e^{\sigma_F \cdot X_{t_k} - \frac{1}{2}\Sigma^2_{t_k}} - K \big)$ is differentiable and its derivative does not depend on $F_{0, t_k}$. Furthermore,
\begin{align*}
\Big|\frac{\partial}{\partial F_{0, t_k}} \Big(q_{k} \cdot \big(F_{0, t_k} \cdot e^{\sigma_F \cdot X_{t_k} - \frac{1}{2}\Sigma^2_{t_k}} - K \big) \Big)\Big| = \Big|q_{k} \cdot e^{\sigma_F \cdot X_{t_k} - \frac{1}{2}\Sigma^2_{t_k}} \Big| \le q_{\max} \cdot e^{\sigma_F \cdot X_{t_k}} \in \mathbb{L}_\mathbb{R}^1(\mathbb{P}).
\end{align*}

\noindent
Then thanks to Lebesgue theorem to interchange derivation and integral, the function $f$ is continuously differentiable in all $F_{0, t_k}$. Likewise one may also show that for all $k = 0, \ldots, n-1$ the function $f$ is continuously differentiable in $q_{k}$, and one may use the envelope theorem. For any $k=0,\ldots,n-1$ we introduce the Lagrangian of the problem \eqref{swing_init_price}
\begin{align*}
	\mathcal{L}\left((q_{k})_{0 \le k \le n-1},\mu, \left(F_{0, t_k}\right)_{0 \le k \le n-1} \right) &= \mathbb{E}\left(\sum_{k = 0}^{n-1} q_{k} \cdot \left( F_{0, t_k} \cdot e^{\sigma_F \cdot X_{t_k} - \frac{1}{2}\Sigma^2_{t_k}} - K \right) \right) \\
	&+ \langle \mu, g\big((q_{k})_{0 \le k \le n-1}, (F_{0, t_k})_{0 \le k \le n-1} \big) \rangle,
\end{align*}

\noindent
where $g = \left(g_1, \ldots, g_{2n+2}\right)^\top$ and $\mu \in \mathbb{R}^{2n+2}$. Then it follows from envelope theorem that for any $k = 0, \ldots, n-1$:
\begin{align*}
\frac{\partial P_0}{\partial F_{0, t_k}} = \mathbb{E}\left(\sum_{k = 0}^{n-1} q_{k}^{*} \cdot \frac{\partial}{\partial F_{0, t_k}} \left( F_{0, t_k} \cdot e^{\sigma_F \cdot X_{t_k} - \frac{1}{2}\Sigma^2_{t_k}} - K \right) \right) = \mathbb{E}\left(\sum_{k = 0}^{n-1} q_{k}^{*} \cdot e^{\sigma_F \cdot X_{t_k} - \frac{1}{2}\Sigma^2_{t_k}} \right).
\end{align*}

\noindent
This completes the proof. 
\end{proof}

We have outlined the theoretical framework for pricing swing contracts. To obtain practical numerical solutions, we introduce some parametric methods based on a global optimization approach.

\section{Swing pricing: a global optimization approach}
\label{sec2}
\indent
In this paper, we propose to solve the optimization problem \eqref{swing_init_price} through a global optimization approach. It is worth noting that the use of that kind of approaches in the swing pricing context is not new and has demonstrated its advantages when compared to dynamic programming based methods like the classic method of Longstaff and Schwartz (see \cite{Longstaff2001ValuingAO}). We refer the reader to the paper of Barrera-Esteve et al. \cite{BarreraEsteve2006NumericalMF} where a global optimization approach had been used to evaluate a swing contract in case with penalties.

The strategy is \emph{bang-bang} so at each exercise date $t_k$, given a cumulative consumption $Q_k$, the random variable $q_k$ takes values in a 2-state set $\{ A_k^{-}(Q_k), A_k^{+}(Q_k) \}$. Since this random variable $q_k$ is measurable with respect to $\mathcal{F}_{t_k}$, there exists $B_k \in \mathcal{F}_{t_k}$ such that
\begin{equation*}
q_k = A_k^{-}(Q_k) + \big(A_k^{+}(Q_k) - A_k^{-}(Q_k) \big) \mathbf{1}_{B_k}, \quad 0 \le k \le n-1.
\end{equation*}
Our approach involves approximating the decision function $\mathbf{1}_{B_k}$ with a parametric function that takes values in the range $[0,1]$. More precisely, we use at each time step $t_k$ a parametric function
\begin{align} \label{def_chi_k}
\chi_k: \mathbb{R}^d \times \mathbb{R}_+ \times \Theta & \to \mathbb{R}, \\
(S, Q, \theta) &\mapsto \chi_k(S, Q, \theta) \notag
\end{align}
where $\Theta$ is a finite-dimensional parameter space that will be specified later. We compose this parametric function with the logistic function $\sigma$ to map $\mathbb{R}$ to the interval $[0,1]$. The parametric strategy writes then 
\begin{equation}
\label{q_parameteric}
q_k(S_k, Q_k^\theta, \theta) = A_k^{-}(Q_k^\theta) + \big(A_k^{+}(Q_k^\theta) - A_k^{-}(Q_k^\theta) \big) \sigma \big( \chi_k(S_{t_k}, Q_k^\theta, \theta) \big), \quad 0 \le k \le n-1,
\end{equation}
where the cumulative consumption $Q_k^{\theta}$ is updated as follows
\begin{equation*}
    Q_0^{\theta} = 0 \quad \text{and} \quad 
    Q_k^{\theta} = \sum_{\ell = 0}^{k-1} q_\ell(S_{t_\ell}, Q_\ell^\theta, \theta) \quad 
    \text{for } 1 \le k \le n-1.
\end{equation*}

Using this parametric strategy, the optimization problem  \eqref{swing_init_price} using the global revenue \eqref{obj_func_J} is approximated by the following finite-dimensional optimization problem
\begin{equation}
    \label{def_opt_pb}
    \theta^* \in \displaystyle \operatorname{argmax}_{\theta \in \Theta} 
    \bar{\mathcal{J}}(\theta) \quad \text{with} \quad
    \bar{\mathcal{J}}(\theta) = \mathbb{E} \left( 
    \sum_{k=0}^{n-1} \psi_k\big(q_k(S_{t_k}, Q_k^\theta, \theta), S_{t_k} \big) \right).
\end{equation}
The function $\bar{\mathcal{J}}$, which corresponds to the reward function following a parametric strategy, is typically unknown, and we will employ stochastic optimization algorithms to maximize it. This initial optimization step will sometimes be referred to as a learning or training phase later on. We will detail the algorithms used in Section~\ref{training_part}.

After completing the learning of the parameter $\theta^*$, we implement a Monte Carlo estimator using the simulation of new scenarios, independent of those used in the stochastic optimization procedure. We then simulate $M_e$ paths $\big( S_{t_k}^{[m]} \big)_{0 \le k \le n-1, 1 \le m \le M_e}$
and the swing contract price is given by 
\begin{equation}
\label{fwd_est_price}
     \widehat{P}_0 = \frac{1}{M_{e}} \sum_{m = 1}^{M_{e}} 
     \sum_{k = 0}^{n-1} 
     \psi_k\Big(q_k \big(S_{t_k}^{[m]}, Q_k^{\theta^*}, {\theta^*}\big), S_{t_k}^{[m]} \Big),
\end{equation}
with the associated confidence interval. We refer to this second step as the evaluation phase.

\subsection{Explicit Payoff-Volume parameterization (\textit{PV strat})}
A first choice of parametric function $\chi_k$ introduced in \eqref{def_chi_k} is defined as follows. Let $\Theta = \mathcal{M}(n \times 3)$ be the set of all $n$-by-$3$ real matrices. Given $\theta \in \Theta$, for $0 \le k \le n-1$ we denote by $\theta_k = (\theta_{k1}, \theta_{k2}, \theta_{k3}) \in \mathbb{R}^3$ the $k+1$-th row of matrix $\theta$. At each timestep $t_k$, the parametric function $\chi_k$ is defined by 
\begin{equation} \label{def_chi_k_PV}
    \chi_k(S, Q, \theta) = \theta_{k1} (S - K) + \theta_{k2} \eta(Q) + \theta_{k3}, 
    \quad \text{with} \quad 
    \eta(Q) = \frac{Q-Q_{\min}}{Q_{\max}-Q_{\min}}.
\end{equation}
To condense this notation, we use the dot product in $\mathbb{R}^3$ and write $\chi_k(S, Q, \theta) = \big\langle \theta_k, I\big\rangle$ with $I = (S-K, \eta(Q), 1)$. The function $\eta$ computes the normalized remaining purchasing capacity. The choice of the function $\eta$ is discussed in Remarks \ref{rq1} and \ref{rq2}. In case $Q_{\min} = Q_{\max}$, one may use $\eta(Q) =\frac{Q-Q_{\min}}{Q_{\min}}$.

The idea behind the representation (2.5) is the following. We believe that the optimal decision/control depends on two key features: the payoff $S-K$ (short-term feature) and, the remaining purchasing capacity (long-term feature) due to global constraints. This remaining purchasing capacity is modelled by function $\eta$ which importance is discussed in Remark \ref{rq1}. We opt for this linear representation because in this setting, coefficients $\theta_{k1},\theta_{k2}$ carry a straightforward interpretation, denoting the relative significance of the corresponding feature (either the payoff $S-K$ or the remaining purchasing capacity $\eta(Q)$) in the explanation of optimal control. For instance, at an exercise date $t_k$, if the volume $Q$ satisfies the global constraints, then the optimal purchasing decision is mostly dependent on the payoff, meaning $\theta_{k1}$ needs to be high enough. Indeed, in this case, the optimal decision is to purchase the maximum possible volume in case the payoff is positive i.e. $S > K$ and the minimum possible volume otherwise. Conversely, if the global constraints are not met at that date, even with a negative payoff, one might be compelled to purchase the maximum allowed volume to adhere to the global constraints, emphasizing the importance of the feature $\eta(Q)$ to the detriment of the payoff. In summary, representation (2.5) provides a simple yet intuitive way to capture the heuristic described above. Besides, it allows to reduce the search space by imposing a particular shape that we believe to be optimal (refer to the previous section for details). Later on, we will see that this simple parametric function~\eqref{def_chi_k_PV} is flexible enough to yield good numerical results.

\begin{remark}[Importance of function $\eta$]
\label{rq1}
At first glance, the function $\eta$ may be seen as a simple normalization of the cumulative consumption. However, it should be noted that several other normalizations were tested and did not lead to a good strategy. We tested: $\eta(Q) = Q$, $\frac{Q - Q_{\min}}{q_{\max} - q_{\min}}$ and some normalizations depending on the step $k$, for example $\frac{k \cdot q_{\max} - Q}{q_{\max} - q_{\min}}$ and $\frac{Q - Q^{down}(t_{k+1})}{Q^{up}(t_{k+1}) - Q^{down}(t_{k+1})}$.
\end{remark}

\begin{remark}[Remaining capacity \textit{Versus} Cumulative Consumption]
\label{rq2}
We observed that the remaining purchasing capacity is a more crucial factor than the current cumulative consumption when determining the volume to buy at each exercise date. This contrasts with the neural network approach tested in \cite{BarreraEsteve2006NumericalMF}, which uses the current cumulative consumption. We have noticed that when using $\eta(Q) = Q$, our algorithm tended to purchase the maximum possible volume at the contract's outset, as long as the payoff was positive. This behaviour suggests that the algorithm did not take into account the fact that purchasing capacity diminishes as the contract approaches its expiration, due to global constraints.
\end{remark}

\begin{remark}[Complexity]
    \label{complexity_pv_strat}
    Using this parameterization, it is worth noting that the finite-dimensional optimization problem~\eqref{def_opt_pb} in the training phase corresponds to a space of dimensionality $3n$.
\end{remark}

\subsection{Neural Network parameterization (\emph{NN strat})}
\label{nn_params}

One way to improve the previous parameterization is to replace each parameter $\theta_k \in \mathbb{R}^3$ (in \emph{PV strat}) with the output of one artificial neural network that depend on some information available at time $t_k$. More precisely, we use a feedforward neural network (FNN) denoted by $\Phi$ and defined by the composition of affine functions and activation functions. Let us be more precise. Consider a network depth of $I \ge 1$ and $\underline{\mathrm{d}} := \min(d_1, \dots, d_{I}) $, where for $i=1,\ldots, I$, the positive integer $d_i$ denotes the number of neurons in the $i^{th}$ layer. The FNN $\Phi$, defined from $\mathbb{R}^d$ to $\mathbb{R}^p$ (in our case $p = 3$), is written as follows
\begin{equation} \label{def_fnn}
    x \in \mathbb{R}^d \mapsto\Phi_{\underline{\mathrm{d}}} (x; \theta) = 
    a_I^{\theta_I} \circ \phi \circ a_{I-1}^{\theta_{I-1}} 
    \circ \cdots \circ \phi \circ a_1^{\theta_1}(x) \in \mathbb{R}^p,
\end{equation}
where 
\begin{itemize}
\item $\theta := \big(\theta_1, \dots, \theta_I \big) \in \Theta := \prod_{i=1}^I \big(\mathcal{M}(d_i \times d_{i-1}) \times \mathbb{R}^{d_i}\big)$ with $d_0 = d$ and $d_I = p$.
\item for each layer $i=1,\dots,I$ the affine function $a_i^{\theta_i}:\mathbb{R}^{d_{i-1}} \to \mathbb{R}^{d_i}$ is defined by $a_i^{\theta_i}(x) = W_i x + b_i$ where $\theta_i = \big(W_i, b_i\big) \in \mathcal{M}(d_i \times d_{i-1}) \times \mathbb{R}^{d_i}$.
\item $\phi:\mathbb{R} \to \mathbb{R}$ is an activation function, which is applied component-wise to a given vector. In the numerical experiments, we only use ReLU (Rectified Linear Unit) function i.e. $x \mapsto \max(x, 0)$.
\end{itemize}

FNNs of form \eqref{def_fnn} have a powerful approximation capacity which is stated by the Universal Approximation Theorem \cite{Cybenko1989ApproximationBS, Hornik1989MultilayerFN}. Among others, a version of the latter theorem (from \cite{LapeyreLelong+2021+227+247}) is reproduced below (see Theorems \ref{uat1} and \ref{uat2}) for the reader convenience. Besides, some quantitative error bounds have been studied in the literature \cite{ATTALI19971069, baron_approx_2, baron_approx} with a quite general review in \cite{devore2020neural}. To state the universal approximation theorem, let us consider the following class of FNNs of form \eqref{def_fnn} with depth $I$ and activation function $\phi$,
\begin{equation}
\label{set_fnn}
    \mathcal{NN}_{I, \underline{d}} := \big\{\Phi_{\underline{\mathrm{d}}} (\cdot; \theta) : \mathbb{R}^d \to \mathbb{R}, \hspace{0.2cm} \theta \in \Theta \big\}
\end{equation}

\noindent
and denote by $\displaystyle \mathcal{NN}_{I, \infty} := \bigcup_{\underline{d} \ge 1} \mathcal{NN}_{I, \underline{d}}$. Then, we have the following theorems giving, at least theoretically, what feedforward neural networks can learn.

\begin{theorem}
\label{uat1}
    Assume that the activation function $\phi$ is non-constant and bounded. Let $\mu$ denote a probability measure on $\mathbb{R}^d$, then for any $I \ge 2$, $\mathcal{NN}_{I, \infty}$ is dense in the space $\mathbb{L}_{\mu}^2\big(\mathbb{R}^d\big)$ of squared $\mu$-integrable functions from $\mathbb{R}^d$ to $\mathbb{R}$.
\end{theorem}

\begin{theorem}
\label{uat2}
    Assume that the activation function $\phi$ is a non constant, bounded and continuous function, then, when $I = 2$, $\mathcal{NN}_{I, \infty}$ is dense into the space of continuous functions from $\mathbb{R}^d$ to $\mathbb{R}$ for the topology of the uniform convergence on compact sets.
\end{theorem}

These theorems state that FNNs of form \eqref{def_fnn} with one hidden layer can approximate a wide class of functions for a given accuracy, whence their characterization as \emph{universal approximators}.

\vspace{0.2cm}
In this second approach based on neural networks, we replace each $\theta_k = (\theta_{k1}, \theta_{k2}, \theta_{k3}) \in \mathbb{R}^3$ of~\eqref{def_chi_k_PV} by the output of a $\mathbb{R}^3$-valued feedforward neural network $\Phi:\mathbb{R}^d \to \mathbb{R}^3$ of form~\eqref{def_fnn}. More precisely, we consider 
\begin{equation} \label{def_chi_k_NN}
    \chi_k(S, Q, \theta) = \big\langle \Phi(t_k, S-K, \eta(Q); \theta), I \big\rangle,
    \quad \text{with} \quad 
    I = (S-K, \eta(Q), 1),
\end{equation}
and $\eta(Q) = \frac{Q-Q_{\min}}{Q_{\max}-Q_{\min}}$. We represent in Figure~\ref{nn_representation} an illustration of the architecture of the neural network.

\vspace{0.1cm}
The concept driving this modelling approach remains centered on modelling the contributions of specific features, as in \eqref{def_chi_k_PV}. However, in this iteration, these contributions are modelled as output of a neural network. This approach enables the integration of non-linearity in the modelling of these contributions, contrasting with \eqref{def_chi_k_PV}, thus enlarging the spectrum of functions to explore. Besides, it is worth noting that this modelling into two steps is mainly due to the fact that we have a priori knowledge of the optimal control's shape: the optimal strategy is bang-bang and we use a sigmoid function as a smooth approximation. Without such prior knowledge, it would be more relevant to model the optimal control using a real-valued neural network. However, in our specific setting where we have a priori knowledge of the optimal control's shape, it turns out that our modelling is more effective than directly modelling the optimal control as the output of a real-valued neural network. The latter means that, given a small number of iterations (1000 iterations in our experiments), our method achieves a significantly better maximum (i.e. higher prices). This is probably because, learning the optimal control from scratch is likely to take more time than starting from a particular shape based on prior knowledge.

\begin{figure}[!h]
    \center
    \includegraphics[scale=0.4]{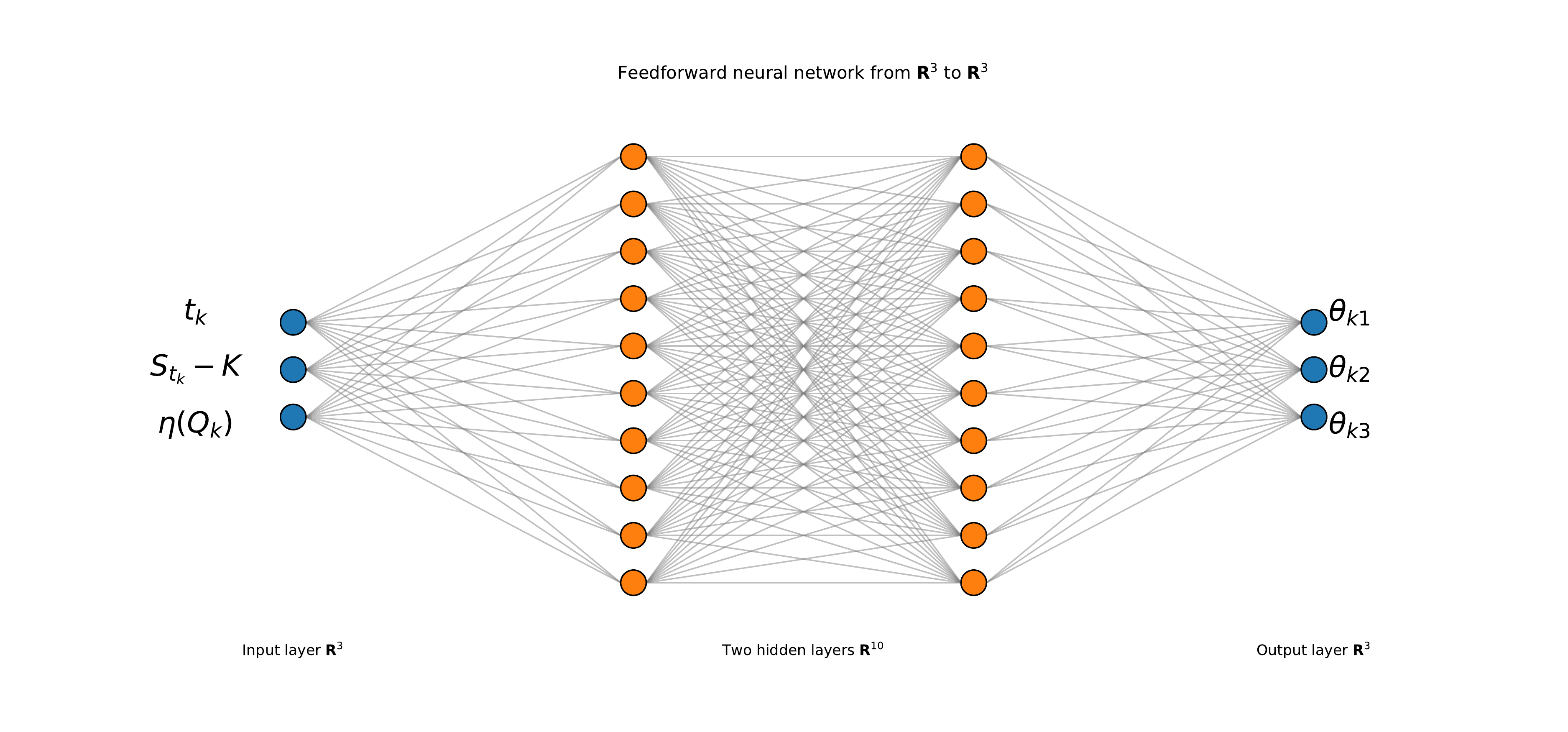}
    \caption{Illustration of feedforward neural network architecture.}
    \label{nn_representation}
\end{figure}

\begin{remark}[Complexity]
    With this parameterization, the parameter space involved in training phase (see ~\eqref{def_opt_pb}) has a dimension equal to $\sum_{i = 1}^{I} d_id_{i-1} + d_i$ which does not depend on the number of exercise dates $n$, unlike the preceding parameterization PV strat where the parameter space dimension increases linearly with the number of exercise dates (as pointed out in Remark \ref{complexity_pv_strat}).
\end{remark}

\section{Stochastic optimization for learning phase}
\label{training_part}

In this section, we will review some recent and efficient stochastic optimization algorithms for solving the high-dimensional problem~\eqref{def_opt_pb}. To establish our notation, let's start by revisiting stochastic gradient ascent for maximizing the deterministic function $\bar{\mathcal{J}}$, which can be expressed as an expectation $\bar{\mathcal{J}}(\theta) = \mathbb{E}\big[ J(\theta, Z) \big]$ with $J(\theta, Z) = \sum_{k=0}^{n-1} \psi_k\big(q_k(S_{t_k}, Q_k^\theta, \theta), S_{t_k} \big)$ and the notation $Z := \big(S_{t_0}, \ldots, S_{t_{n-1}}\big)$. Adhering to the standard notations of stochastic approximation (refer to~\cite{Kushner2003StochasticAA, Robbins2007ASA}), let us define $H:\Theta \times \mathbb{R}^{n}$ as the gradient of $-J$ with respect to $\theta$. We aim to find a root of $h(\theta) = \mathbb{E}\big[H(\theta, Z)\big]$. In practice, we employ automatic differentiation to compute the local gradient $H$
\begin{equation}
\label{gradient_H}
    H(\theta, Z) = -\nabla_\theta J(\theta, Z) = 
    -\sum_{k=0}^{n-1} \nabla_\theta \psi_k\big(q_k(S_{t_k}, Q_k^\theta, \theta), S_{t_k} \big) \hspace{0.4cm}  \text{with} \hspace{0.4cm} Z:= \big(S_{t_0}, \ldots, S_{t_{n-1}} \big).
\end{equation}
The standard stochastic gradient descent to approximate $\theta^*$ a zero of $h$ writes then 
\begin{equation}
    \label{base_sgd_iter}
    \theta_{k+1} = \theta_k - \gamma_{k+1} H\big(\theta_k, Z^{[k+1]}\big), 
\end{equation}
with $(Z^{[k]})_{k\ge 1}$ an \emph{i.i.d.} sequence of copies of $Z = (S_{t_0},\dots, S_{t_{n-1}})$ and $(\gamma_n)_{n \ge 0}$ a real sequence decreasing to 0 satisfying the two conditions 
\begin{equation}
    \label{hyp_DS}
    \sum_{n\ge0}^{} \gamma_n = +\infty \hspace{0.3cm} \text{and} \hspace{0.2cm} \sum_{n\ge0}^{} \gamma_n^2 < +\infty.
\end{equation}
The sequence $(\gamma_n)_{n \in \mathbb{N}}$ is called the algorithm step or learning rate and has to be chosen carefully. It is often taken decreasing by plateau. Practitioners often use the mini-batch version of \eqref{base_sgd_iter}. That is, the local gradient in \eqref{base_sgd_iter} computed over a single simulation is replaced with an average over a small number $B$ of simulations i.e.,
\begin{equation}
    \label{base_sgd_mini_batch_iter}
    \theta_{k+1} = \theta_k - \frac{\gamma_{k+1}}{B} \sum_{b = 1}^{B} H\big(\theta_k, Z^{[k+1], b}\big), 
\end{equation}

\noindent
with $(Z^{[k], b})_{k\ge 1, 1 \le b \le B}$ an \emph{i.i.d.} sequence of copies of $Z = (S_{t_0},\dots, S_{t_{n-1}})$.

\vspace{0.2cm}
The convergence of SGD had been widely studied in the literature. For Lipschitz continuous convex objective functions, some convergence results may be found in \cite{nemirovskiĭ1983problem, Zinkevich2003OnlineCP}. When both smoothness and strong convexity of the objective function hold, a convergence analysis is provided in \cite{NIPS2011_40008b9a}. We refer the reader to \cite{JMLR:v21:19-636, Lei2019StochasticGD} for the non-convex setting and to  \cite{Shamir2012StochasticGD} for non-smooth case.

Several different variants of the standard procedure \eqref{base_sgd_iter} have been developed. In this paper, we use and compare two of them namely, the Adaptive Moment Estimation (Adam) algorithm and the Preconditioned Stochastic Gradient Langevin Dynamics (PSGLD) algorithm.

\begin{remark}
    Hereafter, we will keep the learning rate constant over the training phase i.e. $\gamma_k = \gamma$ as often set by practitioners.
\end{remark}

\subsection{Adaptive Moment Estimation (Adam)}

\indent
Adam \cite{Kingma2015AdamAM} is an efficient stochastic optimization method that only requires first-order gradients with little memory requirement. The method computes individual adaptive learning rates for
different parameters from estimates of first gradient moments. Practically, this is done by means of a preconditioning matrix $P$ (see updating \eqref{adam_iter}). Adam can be seen as a combination of RMSprop \footnote{Unpublished optimization algorithm designed for neural networks, first proposed by Geoffrey Hinton in a Coursera course \url{https://www.cs.toronto.edu/~tijmen/csc321/slides/lecture_slides_lec6.pdf}} and Stochastic Gradient Descent with momentum \cite{sgdMoment}. It replaces the gradient in the stochastic procedure \eqref{base_sgd_iter} with the (scaled) moving average of past gradients values until the considered iteration. The procedure is the following:
\begin{equation}
\label{adam_iter}
\theta_{n+1} = \theta_n - \gamma P_{n+1} \cdot \widehat{M_{n+1}},
\end{equation}

\noindent
where $\widehat{M_{n+1}} $ is a scaled moving average of local gradients $H$ \eqref{gradient_H} and $(P_{n})_n$ is a sequence of preconditioning (diagonal) matrix where diagonal entries are inverse of the scaled moving average of squares of gradients $H$ until iteration $n$ (see Algorithm \ref{algo_adam}). In Algorithm \ref{algo_adam} (and also in Algorithm \ref{algo_psgld}), we adopt the following notations. Given a gradient $g = \big(g_1, \ldots, g_d \big) \in \mathbb{R}^d$, the square component-by-component of $g$ denoted $g \odot g$ is defined by:
\begin{equation}
\label{sq_elt_grad}
g \odot g = \big(g_1^2, \ldots, g_d^2\big).
\end{equation}

\noindent
In the same spirit, for some matrix $A = \big(a_{i, j} \big)_{1 \le i, j \le d}, B = \big(b_{i, j} \big)_{1 \le i, j \le d}$ where $b_{i, j} \neq 0$ for all $1 \le i, j \le d$, we define a component-by-component division as follows,
\begin{equation}
\label{div_elt_grad}
A \oslash B = \big( c_{i, j} \big)_{1 \le i, j \le d} \hspace{0.6cm} \text{where} \hspace{0.2cm} c_{i, j} = a_{i, j} / b_{i, j}.
\end{equation}

\begin{algorithm}
\footnotesize
    \SetKwFunction{isOddNumber}{isOddNumber}
    \SetKwInOut{KwIn}{Input}
    \SetKwInOut{KwOut}{Output}
    \SetKwRepeat{Do}{do}{while}%

    \KwIn{
        \begin{itemize}[noitemsep]
            \item Step size $\gamma$.
            \item $\mu_1, \mu_2 \in [0, 1)$: exponential decay rates for the moment estimates.
            \item Initialize parameter vector $\theta_0$.
        \end{itemize}
    }

    \KwOut{Optimal parameter $\theta^{*}$.}

    \vspace{0.2cm}

    \tcc{Initialize the moving average for the gradient and its square.}

    $M_0 = 0$ (first moment vector), $MS_0 = 0$ (second moment vector).
    
    \vspace{0.2cm}
    
    \While{$\theta_n$ not converged}{
            $\widehat{g}_{n+1} \leftarrow H(\theta_n, Z^{[n+1]}).$ \tcc{Estimate of gradient using \eqref{gradient_H} and automatic differentiation.}

            \vspace{0.2cm}

            $M_{n+1} \leftarrow \mu_1 \cdot M_n + (1 - \mu_1) \cdot \widehat{g}_{n+1}.$ \tcc{Update biased first moment estimate.}

            \vspace{0.2cm}
            
            $MS_{n+1} \leftarrow \mu_2 \cdot MS_n + (1 - \mu_2) \cdot \widehat{g}_{n+1} \odot \widehat{g}_{n+1}.$ \tcc{Update biased second raw moment estimate.}

            \vspace{0.2cm}
            
            $\widehat{M}_{n+1} \leftarrow  M_n / (1 - \mu_1^{n+1}).$ \tcc{Compute bias-corrected first moment estimate.}

            \vspace{0.2cm}
            
            $\widehat{MS}_{n+1} \leftarrow  MS_n / (1 - \mu_2^{n+1}).$ \tcc{Compute bias-corrected second raw moment estimate.}
            
            \vspace{0.2cm}
            
            $P_{n+1} \leftarrow \diag\Big(Id_q \oslash \big(\lambda \cdot Id_d + \sqrt{\widehat{MS}_{n+1}} \big) \Big).$ \tcc{Compute preconditioning matrix.}
            
            \vspace{0.2cm}
            
            $\theta_{n+1} \leftarrow \theta_n - \gamma P_{n+1} \cdot \widehat{M}_{n+1}.$ \tcc{Update parameters using \eqref{base_sgd_iter} or \eqref{base_sgd_mini_batch_iter}.}
        }

    \KwRet{$\theta_n$.}
    \caption{Adam updating. $\lambda$ is a small correction term to avoid division by zero. $Id_d$ denotes the $d \times d$ identity matrix.}

\label{algo_adam}
\end{algorithm}

It is worth noting that, in case of mini-batches, a loop over all mini-batches has to be added in Algorithm \ref{algo_adam} and inside the \emph{while} loop. Adam uses the squared gradients to scale the learning rate like RMSprop and it takes advantage of momentum by using the moving average (driven by $\mu_1, \mu_2$) of the gradient instead of the gradient itself like SGD with momentum. The scaling of the learning rate by the squared gradient allows to solve gradient magnitude problem. Otherwise, sometimes the gradients may be huge and sometimes small, which intricate the choice of a learning rate. The exponential moving average of the gradient allows to \q{denoise} the estimation of the gradient.

\vspace{0.2cm}
Adam has become the default algorithm for stochastic optimization and especially when it comes to train a neural network. Its theoretical convergence had been studied in the literature by its authors \cite{Kingma2015AdamAM} in the convex setting and by other authors in the non-convex setting \cite{Defossez2020ASC, Zou2018ASC}. In this paper, we do not limit ourselves to Adam updating. We propose to explore another algorithm based on Langevin dynamics and which shares, as Adam, good properties when it comes to escape traps. By traps we mean an unstable or repelling critical points  of the loss function to be minimized.

\subsection{Preconditioned Stochastic Gradient Langevin Dynamics (PSGLD)}
Stochastic Gradient Langevin Dynamics (SGLD) has been introduced by Gelfand and Mitter in \cite{GelMit} to minimize a potential $J$ and more recently brought back to light for Bayesian learning~\cite{Welling2011BayesianLV} to approximate a posterior distribution given some data.  The paradigm of SGLD is to inject in Stochastic Gradient  Descent an exogenous  (Gaussian) noise in additive way at each update of the procedure, so that it makes it converge toward an invariant distribution almost concentrated on argmin$\,J$ or, in its simulated annealing version, to converge in probability toward argmin $\,J$. The use of Langevin based optimization algorithms is relatively new in the context of stochastic control problems and had demonstrated its effectiveness in practice \cite{PierreLangevin}. We refer the reader to Appendix \ref{langevin_optim} for theoretical details on the link between Langevin dynamics and optimization problems.

\vspace{0.2cm}
We keep the same notation as in Section \ref{training_part}. Recall that we aim to maximize the function $\bar{\mathcal{J}}(\theta) = \mathbb{E}[J(\theta, Z)]$ and we consider $H : \Theta \times \mathbb{R}^n$ as the gradient of $-J$ with respect to $\theta$ (see \eqref{gradient_H}). After preliminary numerical tests (following see~\cite{bras2022langevin}), we finally implemented  the preconditioned Langevin procedure from the RMSprop procedure (see \cite{Li2015PreconditionedSG}) as follows,
\begin{equation}
    \label{psgld_iter}
    \theta_{n+1} = \theta_n - \gamma P_{n+1} \cdot H\big(\theta_n, Z^{[n+1]}\big) + \sigma \sqrt{\gamma} P_{n+1} \cdot \eta_{n+1},
\end{equation}

\noindent
where,

$\rhd$ $(Z^{[n]})_n$ is a sequence of \emph{i.i.d.} copies of $Z = \big(S_{t_0}, \ldots, S_{t_{n-1}}\big)$.

$\rhd$ $(\eta_n)_n$ is a sequence of \emph{i.i.d.} copies of $\eta \sim \mathcal{N}\big(0, I_d \big)$.

$\rhd$ $(P_n)_n$ is a sequence of diagonal $d \times d$ matrix called \emph{preconditioners} (as in Adam \eqref{adam_iter}) and which entries depend on a moving average of the gradient past values.

\vspace{0.2cm}
The pseudo-code of the preconditioned stochastic gradient langevin dynamic is presented in Algorithm \ref{algo_psgld}, where we consider a constant step setting i.e. $\gamma$ and $\sigma$. From a practical point of view the injection of noise helps the algorithm to escape from traps (see \cite{addgaussnoise}). Besides the convergence of Langevin algorithm has been widely studied. One may refer to \cite{langcvg, durmus2018highdimensional} when the sequence $(\sigma_n)_n$ is constant i.e. $\sigma_n = \sigma$ and to \cite{bras:hal-03891234,Pags2020UnadjustedLA} when the sequence is decreasing.

\begin{algorithm}
\footnotesize
    \SetKwFunction{isOddNumber}{isOddNumber}
    \SetKwInOut{KwIn}{Input}
    \SetKwInOut{KwOut}{Output}
    \SetKwRepeat{Do}{do}{while}%

    \KwIn{
        \begin{itemize}[noitemsep]
            \item Step size $\gamma$.
            \item $\beta \in [0,1)$: exponential decay rate for the moment squared estimates
            \item Initialize parameter vector $\theta_0$.
        \end{itemize}
    }

    \KwOut{Optimal parameter $\theta^{*}$.}

    \vspace{0.2cm}

    \tcc{Initialize the square moving average of the gradient.}

    $MS_0 = 0$ (second moment vector).
    
    \vspace{0.2cm}
    
    \While{$\theta_n$ not converged}{
    \tcc{Compute local gradient using \eqref{gradient_H} and automatic differentiation}
      $\hat{g}_{n+1} \gets H(\theta_n, Z^{[n+1]})$ 

      \tcc{Update biased second raw moment estimate}
      $MS_{n+1} = \beta \cdot MS_n + (1 - \beta) \cdot g_{n+1} \odot g_{n+1}$ 

      \tcc{Compute preconditioner (diagonal) matrix}
      $P_{n+1} \leftarrow \diag \Big(Id_d \oslash \big(\lambda \cdot Id_d + \sqrt{MS_{n+1}} \big)  \Big)$

       \tcc{Update parameters}
      $\theta_{n+1}\gets \theta_n - \gamma P_{n+1}\cdot g_{n+1} +\sigma \sqrt{\gamma}\cdot 
      \mathcal{N}\big(0,\,(P_{n+1}P^{\top}_{n+1})\big)$ 
        }

    \KwRet{$\theta_n$.}
    \caption{PSGLD updating. $\lambda$ is a small correction term to avoid division by zero. $Id_d$ denotes the $d \times d$ identity matrix.}

\label{algo_psgld}
\end{algorithm}

\subsection{Adam \textit{versus} PSGLD}
\label{compar_algo}

\indent

We implement \textit{PV strat} and \textit{NN strat} with both Adam and PSGLD optimization algorithms and in the \textit{case 1} setting. This has been done using \emph{PyTorch} toolbox and with \emph{Nvidia A100-PCIE 40GB} as GPU device. The code can be found using this link: \url{https://github.com/ChristianYeo/swing_parametric_strategies}. As a benchmark, we have implemented the Longstaff-Schwartz method with canonical polynomial functions of degree 3 yielding a price equal to 65.14 ($IC_{95 \%}: [65.08, 65.21]$). The average computation time for the Longstaff-Schwartz method is about 4 minutes when using $10^6$ simulations.

In our experiments, we noticed that with \textit{case 1} (even with \textit{case 2}) setting, regardless the chosen method, the Monte Carlo estimator of the swing price \eqref{fwd_est_price} exhibits a high variance. Thus, for the sake of robust comparison, we adopt the following protocol. For all tables in this section, swing prices are computed by performing \emph{100 times} training phase then evaluation phase (as explained in Section \ref{sec2}) and the 100 resulting price estimators \eqref{fwd_est_price} are averaged. For each of theses 100 replications, the training phase is performed using one batch of size $B = 2^{14} $, a learning rate $\gamma = 0.1$ and in the evaluation phase, we use $M_{e} = 10^6$ simulations. 

\vspace{0.2cm}

In what follows, we denote by $N$ the number of iterations in the stochastic procedure \eqref{base_sgd_mini_batch_iter}, by $\widehat{P_0}$ the (raw) price resulting from \textit{PV strat} or \textit{NN strat} and given by \eqref{fwd_est_price}. $\widehat{P_0}^{BG}$ denotes the price given by forcing the consumption rule (either \eqref{def_chi_k_PV} for \textit{PV strat} or \eqref{def_chi_k_NN} for \textit{NN strat}) to be strictly bang-bang. That is, we shrink the logistic function $\sigma$ in \eqref{q_parameteric} i.e.,
\begin{equation}
    \label{cons_rule_bg_bg}
    \bar{\sigma}(\chi) := \mathbf{1}_{\{\chi \ge 0\}} \in \{0, 1\},
\end{equation}

\noindent
where $\chi$ denotes the decision function (either \eqref{def_chi_k_PV} or \eqref{def_chi_k_NN}). Results are recorded in Tables \ref{results_explicit_param_comp} and \ref{results_psgld_explicit_param_comp} for \textit{PV strat} and in Tables \ref{results_psgld_explicit_param_comp} and \ref{results_psgld_nn_param_comp} for \textit{NN strat}.

\begin{table}[ht!]
\centering
\begin{tblr}{colspec={c},hlines}
\hline
     $N$ &  $\widehat{P_0}$ &  $\widehat{P_0}^{BG}$ & Time (s) \\
     \hline
     3000 & 65.11 ([65.04, 65.18])& 65.15 ([65.08, 65.22]) & 68.4\\
     1000 & 65.02 ([64.95, 65.08])& 65.18 ([65.12, 65.24]) & 21.3\\
\end{tblr}
\caption{Results for \textit{PV strat} using Adam optimization algorithm. A confidence interval (95\%) is in brackets. Column \q{Time} includes both training and valuation times.}
\label{results_explicit_param_comp}
\end{table}

\begin{table}[ht!]
    \centering
\begin{tblr}{colspec={c},hlines}
\hline
     $\sigma$ & $\beta$ & $\widehat{P_0}$ & $\widehat{P_0}^{BG}$ & Time (s) \\
     \hline
     $1 \cdot e^{-6}$ & 0.8  & 65.19 ([65.12, 65.26])& 65.21 ([65.13, 65.28]) & 21.3\\
     $1 \cdot e^{-6}$ & 0.9  & 65.21 ([65.15, 65.28])& 65.23 ([65.17, 65.29]) & 21.3\\
     $1 \cdot e^{-5}$ & 0.9  & 65.22 ([65.16, 65.28])& 65.23 ([65.18, 65.29]) & 21.6\\
\end{tblr}
\caption{Results for \textit{PV strat} using PSGLD optimization algorithm. A confidence interval (95\%) is in brackets. Column \q{Time} includes both training and valuation times. We used $N = 1000$ iterations and $\lambda = 1\cdot e^{-10}$.}
\label{results_psgld_explicit_param_comp}
\end{table}

Tables \ref{results_explicit_param_comp} and \ref{results_psgld_explicit_param_comp} suggest that Adam requires several iterations and, as a result, increases computation time when using a learning rate of 0.1. Whereas, with only $N = 1000$ iterations, using PSGLD optimization algorithm yields a better price. This fact is also observable in \textit{NN strat} (see Tables \ref{results_nn_param_1_comp} and \ref{results_psgld_nn_param_comp}). It seems that adding some noise in the optimization procedure helps to converge quickly.

\begin{table}[ht!]
    \centering
\begin{tblr}{colspec={c},hlines}
\hline
     $N$ &  $\widehat{P_0}$ & $\widehat{P_0}^{BG}$ & Time (s) \\
     \hline
     3000 & 65.22 ([65.16, 65.28])& 65.23 ([65.17, 65.29]) & 150.4\\
     1000 & 65.13 ([65.05, 65.20])& 65.22 ([65.15, 65.29]) & 53.2\\
\end{tblr}
\caption{Results for \textit{NN strat} using Adam optimization algorithm. For the neural network architecture, we used 2 hidden layers ($I = 2$) with 10 units per layer ($q_1 = 10, q_2 = 10$). Values in brackets are confidence intervals (95\%). Columns \q{Time} includes the training and the valuation times.}
\label{results_nn_param_1_comp}
\end{table}

\begin{table}[ht!]
    \centering
\begin{tblr}{colspec={c},hlines}
\hline
     $\sigma$ & $\beta$ & $\widehat{P_0}$ & $\widehat{P_0}^{BG}$ & Time (s) \\
     \hline
      $1 \cdot e^{-6}$ & 0.9 & 65.23 ([65.16, 65.30])& 65.21 ([65.14, 65.28]) & 53.4\\
     $1 \cdot e^{-5}$ & 0.9 & 65.23 ([65.14, 65.31])& 65.22 ([65.13, 65.30]) & 52.3\\
     $1 \cdot e^{-6}$ & 0.8 & 65.27 ([65.20, 65.35])& 65.26 ([65.18, 65.33]) & 51.3\\
\end{tblr}
\caption{Results for \textit{NN strat} using PSGLD optimization algorithm. For the neural network architecture, we used 2 hidden layers ($I = 2$) with 10 units per layer ($q_1 = 10, q_2 = 10$). Values in brackets are confidence intervals (95\%). Columns \q{Time} includes the training and the valuation times. We used $N = 1000$ iterations and $\lambda = 1\cdot e^{-10}$.}
\label{results_psgld_nn_param_comp}
\end{table}

Following the performance of PSGLD updating, we consider this algorithm in the remainder along with the following configuration: $\sigma = 1 \cdot e^{-6}, \beta = 0.8, \lambda = 1 \cdot e^{-10}$. In the latter setting, one may compute sensitivity of the swing price with respect to the initial forward price (see Figure \ref{deltas_forward}).

\begin{figure}[ht]
\centering

\begin{tikzpicture}
\begin{axis}[
	xlabel=Exercise dates,
	ylabel=Delta,
	xmin=0,
	ymin=4,
	grid=both,
	minor grid style={gray!25},
	major grid style={gray!25},
	width=0.5\linewidth,
	height=0.2\paperheight,
	%no marks,
	line width=0.2pt,
	mark size=1.5pt,
	mark options={solid},
	]
\addplot[color=black, mark=*] %
	table[x=time,y=expl,col sep=comma]{datas/greeks/data_30_days.csv};
\addlegendentry{\textit{PV strat}};
\addplot[color=blue, mark=*, style=densely dotted] %
	table[x=time,y=nn,col sep=comma]{datas/greeks/data_30_days.csv};
\addlegendentry{\textit{NN strat}};
\end{axis}
\end{tikzpicture}%
~
\begin{tikzpicture}
\begin{axis}[
	xlabel=Exercise dates,
	ylabel=Delta,
	xmin=0,
	ymin=4,
	grid=both,
	minor grid style={gray!25},
	major grid style={gray!25},
	width=0.5\linewidth,
	height=0.2\paperheight,
	%no marks,
	line width=0.2pt,
	mark size=0.5pt,
	mark options={solid},
	]
\addplot[color=black, mark=*] %
	table[x=time,y=expl,col sep=comma]{datas/greeks/data_365_days.csv};
\addlegendentry{\textit{PV strat}};
\addplot[color=blue, mark=*, style=densely dotted] %
	table[x=time,y=nn,col sep=comma]{datas/greeks/data_365_days.csv};
\addlegendentry{\textit{NN strat}};
\end{axis}
\end{tikzpicture}
\caption{Delta forward in the settings of case 1 (left) and case 2 (right) using \textit{PV strat} and \textit{NN strat}.}
\label{deltas_forward}
\end{figure}
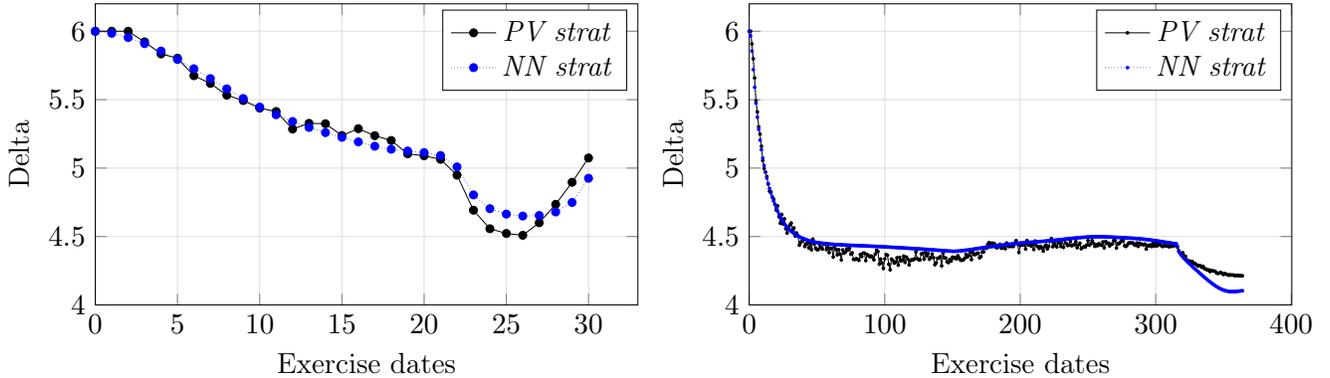

Note that for both optimization algorithms (Adam and PSGLD), different other hyperparameters have been tested and results are reported in Appendix \ref{summary_table_algo}.

\subsection{Practitioner's corner: transfer learning}
\label{transfer_learn}

\indent
We observed that, when the swing contract has several exercise dates or when we use more than one batch per iteration, the execution time increases drastically (see Appendix \ref{summary_table_algo}). To illustrate this point, let us consider a swing contract with an expiry in one year and a daily exercise right (thus 365 exercise dates). Alongside this contract, we consider the diffusion  model \eqref{hjm_model} for the underlying asset with the setting of \textit{case 2}. The results are reported in Table \ref{one_year_swing_baseline}. Notice that, as observed in \cite{BarreraEsteve2006NumericalMF}, the Longstaff-Schwartz method generates numerical instabilities when evaluating this contract at a daily granularity. Moreover, it is worth noting that, we get better prices than optimal quantization method like in \cite{Bardou2009OptimalQF}. 

\vspace{0.2cm}
For all tables in this section, we used one batch of size $B = 2^{14}, M_{e} = 10^8, \gamma = 0.1$ and $N = 1000$ iterations.

\begin{table}[ht]
    \centering
\begin{tblr}{colspec={c},hlines}
\hline
     &  $\widehat{P_0}$ & $\widehat{P_0}^{BG}$ & $(T_{train}, T_{e})$ \\
     \hline
     \textbf{\textit{PV strat}}& 2690.10 ([2689.20, 2691.00])& 2692.59 ([2691.68, 2693.49]) & (274.1, 74.7)\\
     \textbf{\textit{NN strat}} & 2693.74 ([2692.83, 2694.64])& 2694.12 ([2693.21, 2695.02]) & (680.9, 398.3) \\
\end{tblr}
\caption{Results for a one-year swing contract. A confidence interval (95\%) is in brackets. $T_{train}$ denotes the training time and $T_{e}$ the valuation time.}
\label{one_year_swing_baseline}
\end{table}

From Table \ref{one_year_swing_baseline} one may notice that the computation time is not negligible even if, compared to other methods (like Longstaff-Schwartz) it remains reasonable. To reduce it, we propose a method based on the so-called \textit{\textbf{transfer learning}} (for details see \cite{Pan2010ASO}). The latter refers to a machine learning method where a model developed for a task is reused as a starting point for a model on a second and similar task. This method may accelerate the training phase when the swing contract has several exercise dates. In the latter case, we proceed as follows. We first consider an aggregated version. That is, we consider a swing contract with less exercise dates and where the local constraints are aggregated. For instance for a swing contract over one year with daily exercise right and with local constraints $q_{\min} = 0, q_{\max} = 6$, we may consider a contract with one exercise per month (the middle of each month) and where for example, for months with 30 days, the local constraints become $q_{\min} = 0 \times 30 = 0, q_{\max} = 30 \times 6 = 180$. Then we perform the training phase, as explained in Section \ref{sec2}, with a few number of iterations. The resulting parameters are used as an initial guess for the actual pricing problem (with several exercice dates). Note that in the \textit{PV strat} the number of parameters depends on the number of exercise dates. Thus, to initialize the parameters in the problem with 365 exercise dates with the parameters obtained in the aggregated problem (with 12 exercise dates), we assume that all days within the same month behave the same. Therefore the initial guess parameters for all days within the same month are those obtained in the aggregated problem for this month. It turns out that this allows to reach more quickly a descent optimum.

\begin{table}[ht]
    \centering
\begin{tblr}{colspec={c},hlines}
\hline
      & {$\widehat{P_0}$\\$\widehat{P_0}^{BG}$ } & $T_{agg}$ & $(T_{train}, T_{e})$ \\
      \hline
     \textbf{\textit{PV strat}}& {2688.73 ([2687.82, 2689.63]) \\2692.81 ([2691.90, 2693.71]) }&  4.8 & (83.2, 76.3)\\
     \textbf{\textit{NN strat}}& {2693.22 ([2692.31, 2694.13]) \\2694.12 ([2693.21, 2695.02]) }& 6.9 & (206.5, 398.1)\\

\end{tblr}
\caption{Results for a one-year swing contract using transfer learning. We used 500 iterations for the aggregated contract and 300 iterations for the actual one-year contract. $T_{agg}$ denotes the training time for the aggregated contract, $T_{train}$ and $T_{e}$ denote respectively the training time and the valuation time for the actual contract.}
\label{transfer_learning_one_year_swing}
\end{table}

\noindent
We can notice from Table \ref{transfer_learning_one_year_swing} that by means of transfer learning, the training time is reduced by a factor of 3 without degrading the accuracy. Besides, always with in mind the aim of reducing the computation time, transfer learning can be used differently. In a few words, by means of transfer learning, when the market data change or when the contracts settings change, we do not need to start all over again. To illustrate this point, let us consider below three cases: \textit{\textbf{M1}}, \textit{\textbf{M2}}, \textit{\textbf{M3}}, representing some shifts of market/contract settings. Recall that the baseline settings are those considered above (\textit{case 2} setting).

\begin{table}[ht]
    \centering
\begin{tblr}{colspec={c},hlines}
\hline
     \textit{\textbf{Case M1}} & \textit{\textbf{Case M2}} & \textit{\textbf{Case M3}}\\
     \hline
     $F_{0, t_k} = 22$ & $\left(Q_{\min}, Q_{\max} \right) = (1400, 2000)$ & $K = 18$\\
\end{tblr}
\caption{Market data move scenarios}
\label{scenario_mkt_move}
\end{table}

\noindent
In Table \ref{scenario_mkt_move}, we consider three different market move scenarios described as follows. In the first case, we bump the initial forward curve from 20 to 22. In the second case, we change the global constraints from $Q_{\min} = 1300$ and $Q_{\max} = 1900$ to $Q_{\min} = 1400$ and $Q_{\max} = 2000$. In the final case, we reduce the strike price from 20 to 18. We could have changed the local constraints but as already pointed out, swing pricing can always reduce to the case where $q_{\min} = 0, q_{\min} = 1$ (see Appendix \ref{swing decompo}). Now the question is the following: Could our baseline model give accurate prices without training again ? Or could it help us to get accurate prices quickly ? The answers to these questions are recorded in the following tables.

\begin{table}[ht]
    \centering
\begin{tblr}{colspec={c},hlines}
\hline
     Cases & Reuse & Retrain & $(T_{train}, T_{e})$ \\
     \hline
     \textbf{Case \textit{\textbf{M1}}}& {5942.65 ([5941.64, 5943.65]) \\5944.27 ([5943.26, 5945.27]) }& {6005.75 ([6004.70, 6006.79]) \\6008.54 ([6007.50, 6009.58]) } & (85.4, 77.2)\\
     \textbf{Case \textit{\textbf{M2}}}& {2514.86 ([2513.90, 2515.81]) \\2517.26 ([2516.30, 2518.21]) }& {2516.99 ([2516.04, 2517.93]) \\2518.73 ([2517.78, 2519.67]) } & (85.0, 76.9)\\
     \textbf{Case \textit{\textbf{M3}}}& {5655.98 ([5655.06, 5656.89]) \\5657.67 ([5656.75, 5658.58]) }& {5744.91 ([5743.95, 5745.86]) \\5747.39 ([5746.43, 5748.34]) } & (85.9, 77.05)\\
\end{tblr}
\caption{Results for \textit{PV strat}. Column \q{Reuse} provides results when the baseline model parameters are reused as is. Column \q{Retrain} gives results when we retrain with only 300 iterations with the baseline model parameters used as starting values. $T_{train}$ denotes the training time and $T_{e}$ the valuation time. The testing computation time is the same when we reuse the model as when we retrain.}
\label{transfer_learning_mkt_move_strat2}
\end{table}

\begin{table}[ht!]
    \centering
\begin{tblr}{colspec={c},hlines}
\hline
     Cases & Reuse & Retrain & $(T_{train}, T_{e})$ \\
     \hline
     \textbf{Case \textit{\textbf{M1}}}& {5981.17 ([5980.14, 5982.19]) \\5982.57 ([5981.54, 5983.59]) }& {6011.08 ([6010.03, 6012.12]) \\6011.62 ([6010.57, 6012.66]) } & (203.1, 397.8)\\
     \textbf{Case \textit{\textbf{M2}}}& {2509.97 ([2509.01, 2510.90]) \\2511.08 ([2510.12, 2512.03]) }& {2516.92 ([2515.96, 2517.87]) \\2517.58 ([2516.62, 2518.53]) } & (203.7, 397.5)\\
     \textbf{Case \textit{\textbf{M3}}}& {5713.80 ([5712.85, 5714.72]) \\5715.24 ([5714.34, 5716.17]) }& {5749.97 ([5749.01, 5750.92]) \\5750.51 ([5749.55, 5751.46]) } & (203.4, 397.6)\\
\end{tblr}
\caption{Results for \textit{NN strat}. Column \q{Reuse} provides results when the baseline model parameters are reused as is. Column \q{Retrain} gives results when we retrain with only 300 iterations with the baseline model parameters used as starting values. $T_{train}$ denotes the training time and $T_{e}$ the valuation time. The testing computation time is the same whether we reuse the model or we retrain it.}
\label{transfer_learning_mkt_move_strat3}
\end{table}

Tables \ref{transfer_learning_mkt_move_strat2} and \ref{transfer_learning_mkt_move_strat3} demonstrate that in the case where global constraints change (\textit{M2}), using the baseline parameters without additional training gives prices that are comparable to those obtained by training with a few iterations starting from the baseline parameters. Furthermore, it is important to note that increasing the number of iterations beyond 300 (we tested up to 1000 iterations in our experiments) does not lead to better prices than those obtained with 300 iterations (see columns labelled \q{Retrain}). This suggests that reusing the baseline parameters can speed up convergence. This observation remains true regardless of the case (\textit{M1}, \textit{M2} or \textit{M3}).

\section{Numerical experiments}
\label{sec5}

\indent
We now perform additional simulations to demonstrate the effectiveness of the two parameterizations proposed in this paper.

\subsection{Three factor model}

\indent
We first consider a three factor model whose dynamics is given by:
\begin{equation}
	\label{3fac_diff}
	\frac{dF_{t, T}}{F_{t_, T}} = \sigma_{F,1} \cdot e^{-\lambda_1 (T-t)}dW_t^1 + \sigma_{F,2} \cdot e^{-\lambda_2 (T-t)}dW_t^2 + \sigma_{F,3} \cdot e^{-\lambda_3 (T-t)}dW_t^3,
\end{equation}

\noindent
where for all $1 \le i, j \le 3$, the instantaneous correlation is given by:
\begin{equation*}
\langle dW_{\cdot}^i,dW_{\cdot}^j \rangle_t =  \left\{
    \begin{array}{ll}
        dt \hspace{1.4cm} \text{if} \hspace{0.2cm} i = j\\
        \rho_{i, j} \cdot dt \hspace{0.6cm} \text{if} \hspace{0.2cm} i \neq j
    \end{array}
\right.
\end{equation*}

\noindent
In this model, the spot price is given by
\begin{equation*}
S_t = F_{0,t} \cdot \exp\Big(\langle \sigma_F, X_t \rangle - \frac{1}{2}\Sigma_t^2\Big),
\end{equation*}

\noindent
where $\sigma_F = \big(\sigma_{F,1} , \sigma_{F,2} , \sigma_{F,3} )^\top$, $X_t = \big( X_t^1, X_t^2, X_t^3 \big)^\top $ and for all $1 \le i \le 3$,
\begin{equation*}
X_t^i = \int_{0}^{t} e^{-\lambda_i (t-s)}\, \mathrm{d}W_s^i \hspace{0.3cm} \text{and} \hspace{0.2cm} \Sigma_t^2 =  \sum_{1 \le i,j \le 3} \rho_{i, j} \frac{\sigma_{F,i} \sigma_{F,j} }{\lambda_i + \lambda_j} \Big(1 - e^{-(\lambda_i + \lambda_j)t}  \Big).
\end{equation*}

\noindent
We set the following configuration: $\sigma_{F,i}  = 0.7, \lambda_i = 1.5, \rho_{i, j} = \rho \cdot \mathbf{1}_{i \neq j} + \mathbf{1}_{i = j} \in [-1, 1]$ along with \textit{case 1} settings. We use $N = 1000$ iterations, a learning rate $\gamma = 0.1$ and $M_{e} = 1 \cdot e^8$ simulations for the valuation. Results are recorded in Tables \ref{three_factor_model_results_strat1} and \ref{three_factor_model_results_strat2}.

\begin{table}[ht]
    \centering
\begin{tblr}{colspec={c},hlines}
\hline
    $\rho$ & $\widehat{P_0}$ & $\widehat{P_0}^{BG}$ \\
    \hline
    0.6 & 172.71 ([172.51, 172.90])& 172.72 ([172.52, 172.92])\\
	0.3 & 147.79 ([147.63, 147.95])& 147.78 ([147.62, 147.95])\\
	-0.2 & 91.02 ([90.92, 91.12])& 91.02 ([90.92, 91.12])\\
\end{tblr}
\caption{Results using \textit{PV strat}. A confidence interval (95\%) is in brackets. For each result, the execution time (training plus valuation) is roughly equal to 22s.}
\label{three_factor_model_results_strat1}
\end{table}

\begin{table}[ht]
    \centering
\begin{tblr}{colspec={c},hlines}
\hline
    $\rho$ & $\widehat{P_0}$ & $\widehat{P_0}^{BG}$ \\
    \hline
    0.6 & 173.05 ([172.85, 173.24])& 172.98 ([178.78, 173.17])\\
	0.3 & 147.97 ([147.81, 148.14])& 147.92 ([147.75, 148.08])\\
	-0.2 & 91.18 ([91.08, 91.28])& 91.13 ([91.03, 91.23])\\
\end{tblr}
\caption{Results using \textit{NN strat}. A confidence interval (95\%) is in brackets. For the neural network architecture, we used $I = 2$ layers with $q_1 = q_2 = 10$ units. For each result, the execution time (training plus valuation) is roughly equal to 45s.}
\label{three_factor_model_results_strat2}
\end{table}

\newpage

It should be noted that there exists another way to implement \textit{NN strat} which takes into account the structure of forward prices. Specifically, in the three-factor framework \eqref{3fac_diff}, the forward price depends on state variables that are components of the $\mathbb{R}^3$-valued random vector $X_{t_\ell}$. In this context, one can include $X_{t_\ell}$ as an additional input to the neural network described in \eqref{def_chi_k_NN}, in addition to time $t_\ell$, payoff $S_{t_\ell} - K$, and volume normalization $M(Q_\ell)$. That is, $I_\ell = \big(t_\ell, S_{t_\ell} - K, X_{t_\ell}, M(Q_\ell) \big)^\top \in \mathbb{R}^{d+3}$, where $d$ is the dimension of $X_{t_\ell}$ (in the three-factor model, $d = 3$). This approach can help to capture the correlation structure in a multi-factor framework. The prices resulting from this alternative approach are presented in Table \ref{three_factor_model_results_strat2_state_var}.

\begin{table}[ht]
    \centering
\begin{tblr}{colspec={c},hlines}
\hline
    $\rho$ & $\widehat{P_0}$ & $\widehat{P_0}^{BG}$ \\
    \hline
    0.6 & 172.76 ([172.56, 172.95])& 172.70 ([172.50, 172.89])\\
	0.3 & 147.95 ([147.78, 148.11])& 147.91 ([147.75, 148.07])\\
	-0.2 & 91.11 ([91.01, 91.20])& 91.07 ([90.97, 91.17])\\
\end{tblr}
\caption{Results using \textit{NN strat} including state variables. A confidence interval (95\%) is in brackets. For the neural network architecture we used $I = 2$ layers with $q_1 = q_2 = 10$ units. For each result, the execution time (training plus valuation) is roughly equal to 50s.}
\label{three_factor_model_results_strat2_state_var}
\end{table}

It appears that, regardless of whether the model is multi-factor or not, both strategies (\textit{PV strat} and \textit{NN strat}) are able to determine the optimal volume to purchase at each exercise date, based only on the payoff and the cumulative volume. In the following section, we will evaluate the performance of our strategies on a more complex diffusion model.

\subsection{Multi-curve forward diffusion}

\indent
We finally consider a multi-curve model. At a certain valuation date $t$, we consider $p$ risk factors. Each risk factor $i \in \{1,\ldots,p\}$, denoted by $F_{t, T_i}$, is a forward contract observed at date $t$ and expiring at date $T_i$. It is modelled through an instantaneous volatility function $\sigma_t(T_i)$ with a dynamics given by,
\begin{equation}
\label{BGM}
\frac{dF_{t, T_i}}{F_{t, T_i}} = \sigma_t(T_i) dW_t(T_i),
\end{equation}

\noindent
where $(W_t(T_i))_{t \ge 0}$ are standard Brownian motions with instantaneous correlation given by:
$$\langle dW_{\cdot}(T_i), dW_{\cdot}(T_j) \rangle_t = \rho_{i, j} dt.$$

\noindent
The model \eqref{BGM} implies to diffuse one curve per risk factor $F_{t, T_i}$ defined by its maturity date $T_i$. Likewise, the spot at a given date generates its own risk factor. But due to high correlation (at any date, two risk factors are very correlated when maturities are very close), one can see that only a few risk factors really impact the price. In practice, forward contracts delivering in the same month are highly correlated. Thus we will consider that there are as many factors as months. For example, for a contract on a whole year $p = 12$; therefore 12 (correlated) Brownian motions overall. In this model, we consider a valuation date on March $17^{th}$, 2021 and a swing contract delivering on the period January 2022-December 2022. The swing volume constraints are: $q_{\min} = 0, q_{\max} = 1, Q_{\min} = 180, Q_{\max} = 270$. For the diffusion model, the correlation matrix is (\textit{in percentage}):

\setcounter{MaxMatrixCols}{20}

$$\begin{pmatrix}
100 & 99.62 & 98.61 & 91.12 & 91.12 & 91.12 & 90.24 & 90.24 & 90.24 & 87.2 & 87.2 & 87.2\\
99.62 & 100 & 98.91 & 91.65 & 91.65 & 91.65 & 90.68 & 90.68 & 90.68 & 87.8 & 87.8 & 87.8\\
98.61 & 98.61 & 100 & 93.89 & 93.89 & 93.89 & 92.79 & 92.79 & 92.79 & 89.94 & 89.94 & 89.94\\
91.12 & 91.65 & 93.89 & 100 & 100 & 100 & 99.47 & 99.47 & 99.47 & 97.06 & 97.06 & 97.06\\
91.12 & 91.65 & 93.89 & 100 & 100 & 100 & 99.47 & 99.47 & 99.47 & 97.06 & 97.06 & 97.06\\
91.12 & 91.65 & 93.89 & 100 & 100 & 100 & 99.47 & 99.47 & 99.47 & 97.06 & 97.06 & 97.06\\
90.24 & 90.68 & 92.79 & 99.4 & 99.47 & 99.4 & 100 & 100 & 100 & 97.32 & 97.32 & 97.32\\
90.24 & 90.68 & 92.79 & 99.47 & 99.47 & 99.47 & 100 & 100 & 100 & 97.32 & 97.32 & 97.32\\
90.24 & 90.68 & 92.79 & 99.47 & 99.47 & 99.47 & 100  & 100 & 100 & 97.32 & 97.32 & 97.32\\
87.2 & 87.8 & 89.94 & 97.06 & 97.06 & 97.06 & 97.32 & 97.32 & 97.32 &100 & 100& 100\\
87.2 & 87.8 & 89.94 & 97.06 & 97.06 & 97.06 & 97.32 & 97.32 & 97.32 &100& 100 & 100\\
87.2 & 87.8 & 89.94 & 97.06 & 97.06 & 97.06 & 97.32 & 97.32 & 97.32 & 100 & 100 & 100
\end{pmatrix}$$

\vspace{0.3cm}

\noindent
Volatilities and initial forward prices are assumed to be constant per month. For months in 2022, the square of the instantaneous volatility $\sigma^2_t(T_i)$ is given by Figure \ref{vol_curve_bgm} and the initial forward curve is given by Figure \ref{f0_curve_bgm}.

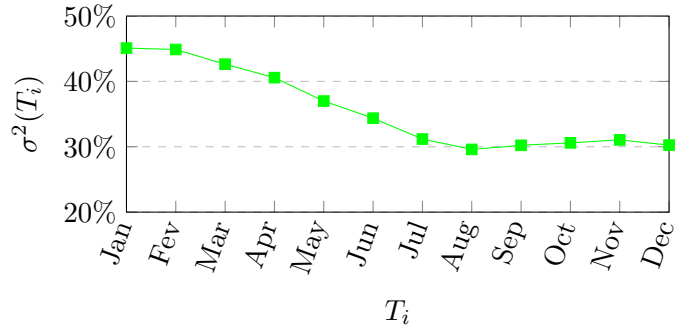
\begin{figure}[ht!]
\centering
\begin{tikzpicture}
\begin{axis}[
    xlabel={$T_i$},
    ylabel={$\sigma^2(T_i)$},
    xmin=0, xmax=11,
    ymin=20, ymax=50,
    ymajorgrids=true,
    grid style=dashed,
   	width=0.5\linewidth,
	height=0.15\paperheight,
	xticklabels={Jan,Fev,Mar, Apr, May, Jun, Jul, Aug, Sep, Oct, Nov,Dec},xtick={0,...,11},
  x tick label style={rotate=70,anchor=east},
  yticklabel = {\pgfmathprintnumber\tick\%}
]

\addplot[
    color=green,
    mark=square*,
    ]
    coordinates {
    (0,45.09)(1,44.88)(2,42.63)(3,40.57)(4,36.99)(5,34.37)(6,31.18)(7,29.59)(8,30.22)(9,30.60)(10,31.05)(11,30.25)
    };
    
\end{axis}

\end{tikzpicture}
\caption{Square of instantaneous volatility per risk factor observed on 17-March-2021. Values are: $\big(45.09\%,44.88\%,42.63\%,40.57\%,36.99\%,34.37\%,31.18\%,29.59\%,30.22\%,30.6\%,31.05\%,30.25\%\big)$}
\label{vol_curve_bgm}
\end{figure}

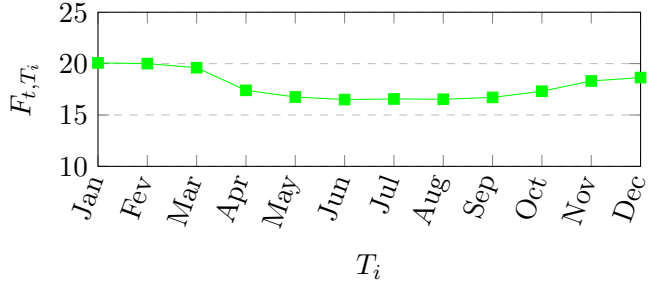
\begin{figure}[ht!]
\centering
\begin{tikzpicture}
\begin{axis}[
    xlabel={$T_i$},
    ylabel={$F_{t, T_i}$},
    xmin=0, xmax=11,
    ymin=10, ymax=25,
    ymajorgrids=true,
    grid style=dashed,
   	width=0.5\linewidth,
	height=0.13\paperheight,
	xticklabels={Jan,Fev,Mar, Apr, May, Jun, Jul, Aug, Sep, Oct, Nov,Dec},xtick={0,...,11},
  x tick label style={rotate=70,anchor=east}
]

\addplot[
    color=green,
    mark=square*,
    ]
    coordinates {
    (0,20.07)(1,20.0)(2,19.6)(3,17.4)(4,16.75)(5,16.5)(6,16.56)(7,16.53)(8,16.71)(9,17.31)(10,18.31)(11,18.64)
    };
    
\end{axis}

\end{tikzpicture}
\caption{Forward prices per risk factor observed on 17-March-2021. Prices per risk factor are: $\big(20.07,20,19.6,17.4,16.75,16.5,16.56,16.53,16.71,17.31,18.31,18.64\big)$}
\label{f0_curve_bgm}
\end{figure}

\noindent
Finally, the strike price corresponds to the average of all initial forward prices; which gives 17.865. For this diffusion model, \textit{PV strat} does not perform well. This is probably due to the multi-curve framework which is more complex than a multifactor model. However, by modifying \textit{NN strat} as in the preceding section by adding, in neural network inputs, Brownian motions $\big(W_{t_k}(T_i)\big)_{1 \le i \le p}$ at each exercise date $t_k$, allows to achieve very good prices. We compare the latter implementation of \textit{NN strat} with a method based on a static replication of swing contracts by means of spread options. The latter method aims at approximating swing options prices by selecting a combination of spread options (\textit{SO}) where the underlying is driven by the same dynamics \eqref{BGM}. The selection is done through a linear programming method. Results are recorded in Table \ref{bgm_strat2_state_var}.

\begin{table}[ht]
    \centering
\begin{tblr}{colspec={c},hlines}
\hline
    $\widehat{P_0}$ & $\widehat{P_0}^{BG}$  & $SO$ \\
    \hline
    708.46 ([705.76, 711.16])& 711.32 ([708.61, 714.03]) & 645.3\\
\end{tblr}
\caption{Results using \textit{NN strat}. A confidence interval (95\%) is in brackets. For the neural network architecture we used $I = 2$ layers with $q_1 = q_2 = 50$ units.}
\label{bgm_strat2_state_var}
\end{table}

\section*{Conclusion}
\indent
We introduced two parametric methods for the pricing of swing contracts based on a global optimization approach. The first method involves an explicit parameterization of the optimal control based on some heuristics, while the second method improves upon the first by using neural networks. We conducted numerical experiments to compare two optimization algorithms (Adam and PSGLD). Our results demonstrate that using Langevin based optimization algorithms, namely PSGLD, allows us to achieve better prices with short computation time. Additionally, we tested our neural network parameterization in various complex diffusion models and demonstrated its robustness and accuracy compared to state-of-the-art methods. We can hence conclude that our neural network based method is well-suited to the pricing swing options.

\section*{Acknowledgements}
The third author is grateful for the financial support provided by Engie Global Markets via CIFRE agreement and would like to thank Asma Meziou and Frederic Muller for throughout insights on swing contracts.

\footnotesize
\appendix

\section{Swing contract decomposition}
\label{swing decompo}

\indent
We aim to prove that any swing contract can be reduced to a normalized contract, namely a swing contract with local constraints $q_{\min} = 0, q_{\max} = 1$. The swing contract price is given by
\begin{equation*}
    P_0 = \esssup_{(q_{\ell})_{0 \le \ell \le n-1} \in \mathcal{A}_{0, 0}^{Q_{\min}, Q_{\max}}} \hspace{0.1cm} \mathbb{E}\left(\sum_{\ell=0}^{n-1} q_{\ell} \times (S_{t_\ell} - K) \right),
\end{equation*}

\noindent
where 
$$\mathcal{A}_{0, 0}^{Q_{\min}, Q_{\max}} = \left\{(q_{\ell})_{0 \le \ell \le n-1}, \hspace{0.1cm} q_{\ell} : (\Omega, \mathcal{F}_{t_\ell}^S, \mathbb{P}) \mapsto [q_{\min}, q_{\max}], \hspace{0.1cm} \sum_{\ell = 0}^{n-1} q_{\ell} \in [Q_{\min}, Q_{\max}] \right\}.$$

\noindent
It follows from the linearity of the expectation that,
\begin{align*}
    P_0 &= q_{\min} \times \mathbb{E}\left(\sum_{\ell=0}^{n-1} (S_{t_\ell} - K) \right) + (q_{\max} - q_{\min}) \times \esssup_{(q_{\ell})_{0 \le \ell \le n-1} \in \mathcal{A}_{0, 0}^{Q_{\min}, Q_{\max}}} \hspace{0.1cm} \mathbb{E}\left(\sum_{\ell=0}^{n-1} \frac{q_{\ell} - q_{\min}}{q_{\max} - q_{\min}} \times (S_{t_\ell} - K) \right).
\end{align*}

\noindent
The first term in the last equality is given by
$$\mathbb{E}\left(\sum_{\ell=0}^{n-1} (S_{t_\ell} - K) \right) = \sum_{\ell=0}^{n-1}\mathbb{E}\left(S_{t_\ell} \right) - n \cdot K,$$

\noindent
and can be easily computed using either closed formula (depending on the underlying diffusion model) or Monte-Carlo method. Let us consider the second term. Let $(q_{\ell})_{0 \le \ell \le n-1} \in \mathcal{A}_{0, 0}^{Q_{\min}, Q_{\max}}$ and define $\Tilde{q}_{\ell} = \frac{q_{\ell} - q_{\min}}{q_{\max} - q_{\min}}$. Note that $(\Tilde{q}_{\ell})_{0 \le \ell \le n-1} \in \mathcal{A}_{0, 0}^{\Tilde{Q}_{\min}, \Tilde{Q}_{\max}}$ where
$$\mathcal{A}_{0, 0}^{\Tilde{Q}_{\min}, \Tilde{Q}_{\max}} = \left\{(q_{\ell})_{0 \le \ell \le n-1}, \hspace{0.1cm} q_{\ell} : (\Omega, \mathcal{F}_{t_\ell}, \mathbb{P}) \mapsto [0, 1], \hspace{0.1cm} \sum_{\ell = 0}^{n-1} q_{\ell} \in [\Tilde{Q}_{\min}, \Tilde{Q}_{\max}] \right\},$$

\noindent
and
$$\Tilde{Q}_{\min} = \frac{(Q_{\min} - n \cdot q_{\min})_{+}}{q_{\max} - q_{\min}} \hspace{0.7cm} \Tilde{Q}_{\max} = \frac{(Q_{\max} - n \cdot q_{\min})_{+}}{q_{\max} - q_{\min}}.$$

\noindent
Thus,
\begin{align*}
    \mathbb{E}\left(\sum_{\ell=0}^{n-1} \frac{q_{\ell} - q_{\min}}{q_{\max} - q_{\min}} \times (S_{t_\ell} - K) \right) &= \mathbb{E}\left(\sum_{\ell=0}^{n-1} \Tilde{q}_{\ell} \times (S_{t_\ell} - K) \right)\\
    &\le \esssup_{(q_{\ell})_{0 \le \ell \le n-1} \in \mathcal{A}_{0, 0}^{\Tilde{Q}_{\min}, \Tilde{Q}_{\max}}}\hspace{0.1cm} \mathbb{E}\left(\sum_{\ell=0}^{n-1} q_{\ell} \times (S_{t_\ell} - K) \right).
\end{align*}

\noindent
Therefore taking the supremum yields,
$$\esssup_{(q_{\ell})_{0 \le \ell \le n-1} \in \mathcal{A}_{0, 0}^{Q_{\min}, Q_{\max}}} \hspace{0.1cm} \mathbb{E}\left(\sum_{\ell=0}^{n-1} \frac{q_{\ell} - q_{\min}}{q_{\max} - q_{\min}} \times (S_{t_\ell} - K) \right) \le \esssup_{(q_{\ell})_{0 \le \ell \le n-1} \in \mathcal{A}_{0, 0}^{\Tilde{Q}_{\min}, \Tilde{Q}_{\max}}} \hspace{0.1cm} \mathbb{E}\left(\sum_{\ell=0}^{n-1} q_{\ell} \times (S_{t_\ell} - K) \right).$$

\noindent
Conversely let $(q_{\ell})_{0 \le \ell \le n-1} \in \mathcal{A}_{0, 0}^{\Tilde{Q}_{\min}, \Tilde{Q}_{\max}}$ and define $\Tilde{q}_{\ell} = q_{\min} + (q_{\max} - q_{\min}) \cdot q_{\ell} \in [q_{\min}, q_{\max}]$. It follows $\displaystyle \sum_{\ell = 0}^{n-1} q_{\ell} \in [Q_{\min}, Q_{\max}]$ so that $\left(\Tilde{q}_{\ell} \right)_{0 \le \ell \le n-1} \in \mathcal{A}_{0, 0}^{Q_{\min}, Q_{\max}}$. Thus,
\begin{align*}
    \mathbb{E}\left(\sum_{\ell=0}^{n-1} q_{\ell} \times (S_{t_\ell} - K) \right) &= \mathbb{E}\left(\sum_{\ell=0}^{n-1} \frac{\Tilde{q}_{\ell} - q_{\min}}{q_{\max} - q_{\min}} \times (S_{t_\ell} - K) \right)\\
    &\le \esssup_{(q_{\ell})_{0 \le \ell \le n-1} \in \mathcal{A}_{0, 0}^{Q_{\min}, Q_{\max}}} \hspace{0.1cm} \mathbb{E}\left(\sum_{\ell=0}^{n-1} \frac{q_{\ell} - q_{\min}}{q_{\max} - q_{\min}} \times (S_{t_\ell} - K) \right).
\end{align*}

\noindent
Taking the supremum, we get,
$$\esssup_{(q_{t_\ell})_{0 \le \ell \le n-1} \in \mathcal{A}_{0, 0}^{\Tilde{Q}_{\min}, \Tilde{Q}_{\max}}} \hspace{0.1cm} \mathbb{E}\left(\sum_{\ell=0}^{n-1} q_{\ell} \times (S_{t_\ell} - K) \right) \le \esssup_{(q_{\ell})_{0 \le \ell \le n-1} \in \mathcal{A}_{0, 0}^{Q_{\min}, Q_{\max}}} \hspace{0.1cm} \mathbb{E}\left(\sum_{\ell=0}^{n-1} \frac{q_{\ell} - q_{\min}}{q_{\max} - q_{\min}} \times (S_{t_\ell} - K) \right).$$

\noindent
Therefore,
\begin{align*}
    P_0 &= q_{\min} \times \mathbb{E}\left(\sum_{\ell=0}^{n-1} (S_{t_\ell} - K) \right) + (q_{\max} - q_{\min}) \times \esssup_{(q_{\ell})_{0 \le \ell \le n-1} \in \mathcal{A}_{0, 0}^{\Tilde{Q}_{\min}, \Tilde{Q}_{\max}}}\hspace{0.1cm} \mathbb{E}\left(\sum_{\ell=0}^{n-1} q_{\ell} \times (S_{t_\ell} - K) \right).
\end{align*}

\section{On the rate of convergence}
\label{sec4}

\indent
This section aims to estimate, at least numerically, the rate of convergence.  We denote by $U_n$ the price obtained after the $n^{th}$ iteration in the training step. We make the assumption that for some constant $C > 0$
$$U_n = U_{\infty} + \frac{C}{n^{\alpha}},$$

\noindent
where $U_{\infty}$ is the limit of the stochastic procedure that we assume to exist. Thus,
$$\log(|U_{2n} - U_{n}|) = K_{\alpha} + \alpha \hspace{0.1cm} \log(\frac{1}{n}), \hspace{0.5cm} \text{where} \hspace{0.1cm} K_{\alpha} = \log(|C(2^{-\alpha} - 1)|).$$

\noindent
Therefore the coefficient $\alpha$ (representing the rate of convergence following the assumption) appears to be the slope in the $\log - \log$ regression of $|U_{2n} - U_{n}|$ against $\frac{1}{n}$.

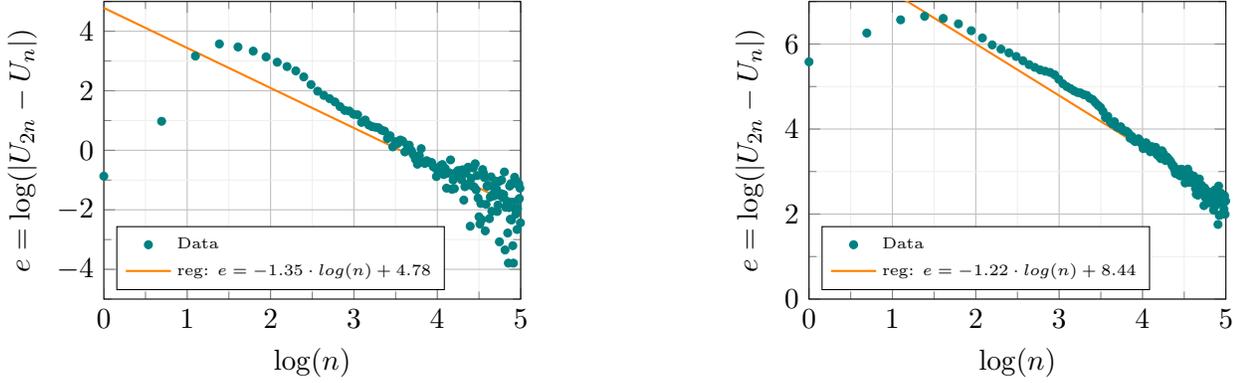
\begin{figure}[ht!]
  \centering
  \begin{subfigure}[b]{0.45\textwidth}
%%%    \begin{adjustbox}{width=\linewidth} % rescale box
\begin{tikzpicture}{scale = 0.4}
\begin{axis}[
    xmin = 0, xmax = 5,
    ymin = -5, ymax = 5,
    width = 0.9\textwidth,
    height = 0.7\textwidth,
    xtick distance = 1,
    ytick distance = 2,
    grid = both,
    minor tick num = 1,
    major grid style = {lightgray},
    minor grid style = {lightgray!25},
    xlabel = {$\log(n)$},
    ylabel = {$e=\log(|U_{2n} - U_n|)$},
    legend cell align = {left},
    legend pos = south west,
    legend style = {font=\tiny},
    mark size=1.5pt,
]
 
% Plot data
\addplot[
    teal, 
    only marks
] table[x = n, y = e, col sep = comma] {datas/rate_cvg/reg_explicit_30_days.csv};
 
% Linear regression
\addplot[
    thick,
    orange
] table[
    x = n,
    y = {create col/linear regression={y=e}},
    col sep = comma
] {datas/rate_cvg/reg_explicit_30_days.csv};
 
% Add legend
\addlegendentry{Data}
\addlegendentry{
    reg: $ e =
    \pgfmathprintnumber{\pgfplotstableregressiona}
    \cdot log(n)
    \pgfmathprintnumber[print sign]{\pgfplotstableregressionb}$
};
 
\end{axis}
\end{tikzpicture}
\end{subfigure}%
\hfill
\begin{subfigure}[b]{0.45\textwidth}
\begin{tikzpicture}{scale = 0.4}
\begin{axis}[
    xmin = 0, xmax = 5,
    ymin = 0, ymax = 7,
    width = 0.9\textwidth,
    height = 0.7\textwidth,
    xtick distance = 1,
    ytick distance = 2,
    grid = both,
    minor tick num = 1,
    major grid style = {lightgray},
    minor grid style = {lightgray!25},
    xlabel = {$\log(n)$},
    ylabel = {$e=\log(|U_{2n} - U_n|)$},
    legend cell align = {left},
    legend pos = south west,
    legend style = {font=\tiny},
    mark size=1.5pt,
]
 
% Plot data
\addplot[
    teal, 
    only marks
] table[x = n, y = e, col sep = comma] {datas/rate_cvg/reg_explicit_365_days.csv};
 
% Linear regression
\addplot[
    thick,
    orange
] table[
    x = n,
    y = {create col/linear regression={y=e}},
    col sep = comma
] {datas/rate_cvg/reg_explicit_365_days.csv};
 
% Add legend
\addlegendentry{Data}
\addlegendentry{
    reg: $ e =
    \pgfmathprintnumber{\pgfplotstableregressiona}
    \cdot log(n)
    \pgfmathprintnumber[print sign]{\pgfplotstableregressionb}$
};
 
\end{axis}
\end{tikzpicture}
\end{subfigure}
\caption{Numerical Convergence (logarithmic scale) for \textit{PV strat}. Case 1 (left), case 2 (right)}
\label{rate_cvg1}
\end{figure}

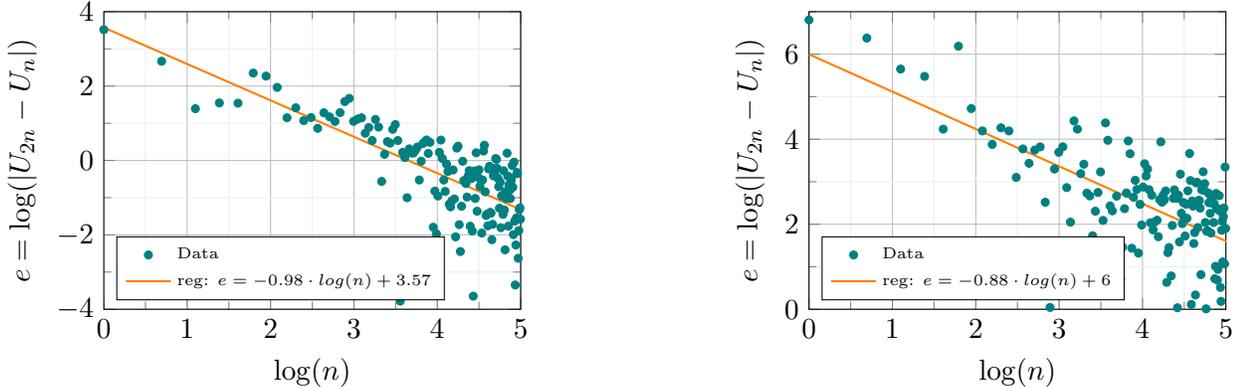
\begin{figure}[ht!]
  \centering
  \begin{subfigure}[b]{0.45\textwidth}
%%%    \begin{adjustbox}{width=\linewidth} % rescale box
\begin{tikzpicture}{scale = 0.4}
\begin{axis}[
    xmin = 0, xmax = 5,
    ymin = -4, ymax = 4,
    width = 0.9\textwidth,
    height = 0.7\textwidth,
    xtick distance = 1,
    ytick distance = 2,
    grid = both,
    minor tick num = 1,
    major grid style = {lightgray},
    minor grid style = {lightgray!25},
    xlabel = {$\log(n)$},
    ylabel = {$e=\log(|U_{2n} - U_n|)$},
    legend cell align = {left},
    legend pos = south west,
    legend style = {font=\tiny},
    mark size=1.5pt,
]
 
% Plot data
\addplot[
    teal, 
    only marks
] table[x = n, y = e, col sep = comma] {datas/rate_cvg/reg_nn_30_days.csv};
 
% Linear regression
\addplot[
    thick,
    orange
] table[
    x = n,
    y = {create col/linear regression={y=e}},
    col sep = comma
] {datas/rate_cvg/reg_nn_30_days.csv};
 
% Add legend
\addlegendentry{Data}
\addlegendentry{
    reg: $ e =
    \pgfmathprintnumber{\pgfplotstableregressiona}
    \cdot log(n)
    \pgfmathprintnumber[print sign]{\pgfplotstableregressionb}$
};
 
\end{axis}
\end{tikzpicture}
\end{subfigure}%
\hfill
\begin{subfigure}[b]{0.45\textwidth}
\begin{tikzpicture}{scale = 0.4}
\begin{axis}[
    xmin = 0, xmax = 5,
    ymin = 0, ymax = 7,
    width = 0.9\textwidth,
    height = 0.7\textwidth,
    xtick distance = 1,
    ytick distance = 2,
    grid = both,
    minor tick num = 1,
    major grid style = {lightgray},
    minor grid style = {lightgray!25},
    xlabel = {$\log(n)$},
    ylabel = {$e=\log(|U_{2n} - U_n|)$},
    legend cell align = {left},
    legend pos = south west,
    legend style = {font=\tiny},
    mark size=1.5pt,
]
 
% Plot data
\addplot[
    teal, 
    only marks
] table[x = n, y = e, col sep = comma] {datas/rate_cvg/reg_nn_365_days.csv};
 
% Linear regression
\addplot[
    thick,
    orange
] table[
    x = n,
    y = {create col/linear regression={y=e}},
    col sep = comma
] {datas/rate_cvg/reg_nn_365_days.csv};
 
% Add legend
\addlegendentry{Data}
\addlegendentry{
    reg: $ e =
    \pgfmathprintnumber{\pgfplotstableregressiona}
    \cdot log(n)
    \pgfmathprintnumber[print sign]{\pgfplotstableregressionb}$
};
 
\end{axis}
\end{tikzpicture}
\end{subfigure}
\caption{Numerical Convergence (logarithmic scale) \textit{NN strat}. Case 1 (left), case 2 (right)}
\label{rate_cvg2}
\end{figure}

\noindent
Figures \ref{rate_cvg1} and \ref{rate_cvg2} suggest that both parameterizations give a rate of convergence of order $\mathcal{O}(\frac{1}{N})$ (where $N$ is the number of iterations). This convergence rate is much faster than Monte Carlo, and suggests that our parameterizations are efficient alternatives to Longstaff-Schwartz method for this problem.

\section{Estimator variance and computation time}
\label{add_res}

\indent
In this section we illustrate the high variance phenomenon which may appear when pricing swing options. To this end, we compute 100 realizations of swing contract price within \textit{case 1} setting for the three methods: \textit{PV strat}, \textit{NN strat} and Longstaff-Schwartz. The distributions of prices are represented in Figures \ref{var_hist_ep}, \ref{var_hist_nn}, \ref{var_hist_ls}.

\begin{figure}[ht!]
  \centering
  \begin{subfigure}[b]{0.2\textwidth}
    %\begin{adjustbox}{width=\linewidth} % rescale box
    \begin{tikzpicture}
\begin{axis}[
    ybar,
    ymin=1,
    ymax = 100,
    xmin = 62,
    xmax = 68,
    width=1.5\linewidth,
    xtick distance=1,
]
\addplot +[
    hist={
        %bins=1,
        %data min=62,
        %data max=68
    }   
] table[y index = 0]{datas/var_est/var_hist_explicit_param_100000.csv};
\end{axis}
\end{tikzpicture}
\end{subfigure}
\hfill
\begin{subfigure}[b]{0.2\textwidth}
    \begin{tikzpicture}
\begin{axis}[
    ybar,
    ymin=1,
    ymax = 100,
    xmin = 62,
    xmax = 68,
    width=1.5\linewidth,
    xtick distance=1,
]
\addplot +[
    hist={
        %bins=1,
        %data min=62,
        %data max=68
    }   
] table[y index = 0]{datas/var_est/var_hist_explicit_param_1000000.csv};
\end{axis}
\end{tikzpicture}
\end{subfigure}
\hfill
\begin{subfigure}[b]{0.2\textwidth}
    \begin{tikzpicture}
\begin{axis}[
    ybar,
    ymin=1,
    ymax = 100,
    xmin = 62,
    xmax = 68,
    width=1.5\linewidth,
    xtick distance=1,
]
\addplot +[
    hist={
        %bins=1,
        %data min=62,
        %data max=68
    }   
] table[y index = 0]{datas/var_est/var_hist_explicit_param_10000000.csv};
\end{axis}
\end{tikzpicture}
\end{subfigure}
\hfill
\begin{subfigure}[b]{0.2\textwidth}
    \begin{tikzpicture}
\begin{axis}[
    ybar,
    ymin=1,
    ymax = 100,
    xmin = 62,
    xmax = 68,
    width=1.5\linewidth,
    xtick distance=1,
]
\addplot +[
    hist={
        %bins=5,
        %data min=62,
        %data max=68
    }   
] table[y index = 0]{datas/var_est/var_hist_explicit_param_500000000.csv};
\end{axis}
\end{tikzpicture}
\end{subfigure}
\caption{Distribution of swing prices using \textit{PV strat}. From left to right we used successively $M_{e} = 10^{5}, 10^{6}, 10^{7}, 5 \times 10^{8}$ simulations.}
\label{var_hist_ep}
\end{figure}
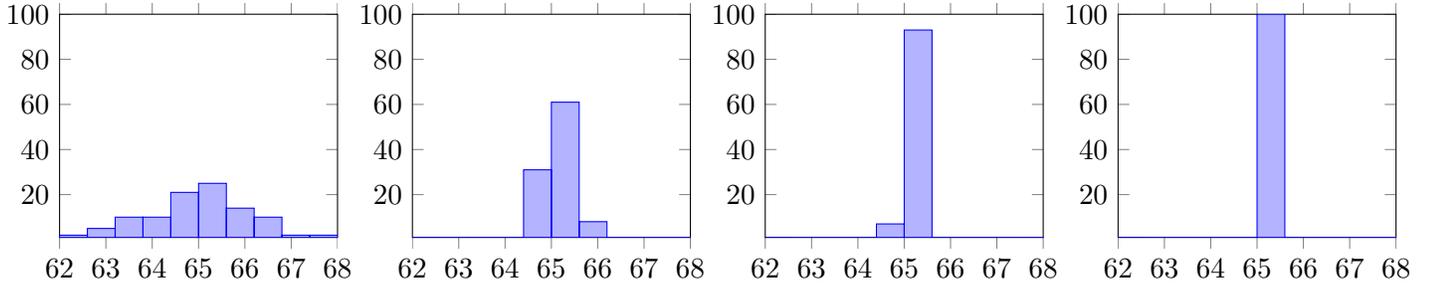

\newpage

\begin{figure}[ht!]
  \centering
  \begin{subfigure}[b]{0.2\textwidth}
    %\begin{adjustbox}{width=\linewidth} % rescale box
    \begin{tikzpicture}
\begin{axis}[
    ybar,
    ymin=1,
    ymax = 100,
    xmin = 62,
    xmax = 68,
    width=1.5\linewidth,
    xtick distance=1,
]
\addplot +[
    hist={
        %bins=5,
        %data min=62,
        %data max=68
    }   
] table[y index = 0]{datas/var_est/var_hist_nn_param_100000.csv};
\end{axis}
\end{tikzpicture}
\end{subfigure}
\hfill
\begin{subfigure}[b]{0.2\textwidth}
    \begin{tikzpicture}
\begin{axis}[
    ybar,
    ymin=1,
    ymax = 100,
    xmin = 62,
    xmax = 68,
    width=1.5\linewidth,
    xtick distance=1,
]
\addplot +[
    hist={
        %bins=2,
        %data min=62,
        %data max=68
    }   
] table[y index = 0]{datas/var_est/var_hist_nn_param_1000000.csv};
\end{axis}
\end{tikzpicture}
\end{subfigure}
\hfill
\begin{subfigure}[b]{0.2\textwidth}
    \begin{tikzpicture}
\begin{axis}[
    ybar,
    ymin=1,
    ymax = 100,
    xmin = 62,
    xmax = 68,
    width=1.5\linewidth,
    xtick distance=1,
]
\addplot +[
    hist={
        %bins=2,
        %data min=62,
        %data max=68
    }   
] table[y index = 0]{datas/var_est/var_hist_nn_param_10000000.csv};
\end{axis}
\end{tikzpicture}
\end{subfigure}
\hfill
\begin{subfigure}[b]{0.2\textwidth}
    \begin{tikzpicture}
\begin{axis}[
    ybar,
    ymin=1,
    ymax = 100,
    xmin = 62,
    xmax = 68,
    width=1.5\linewidth,
    xtick distance=1,
]
\addplot +[
    hist={
        %bins=1,
        %data min=62,
        %data max=68
    }   
] table[y index = 0]{datas/var_est/var_hist_nn_param_500000000.csv};
\end{axis}
\end{tikzpicture}
\end{subfigure}
\caption{Distribution of swing prices using \textit{NN strat}. From left to right we used successively $M_{eval} = 10^{5}, 10^{6}, 10^{7}, 5 \times 10^{8}$ simulations.}
\label{var_hist_nn}
\end{figure}
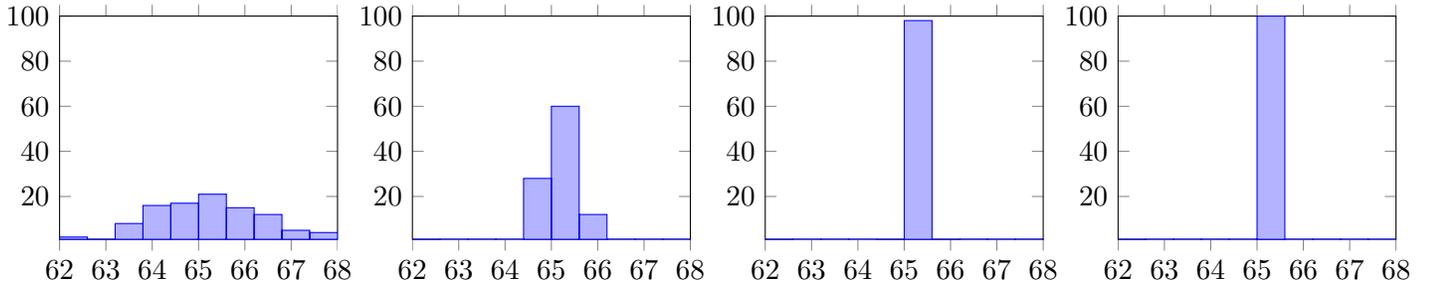

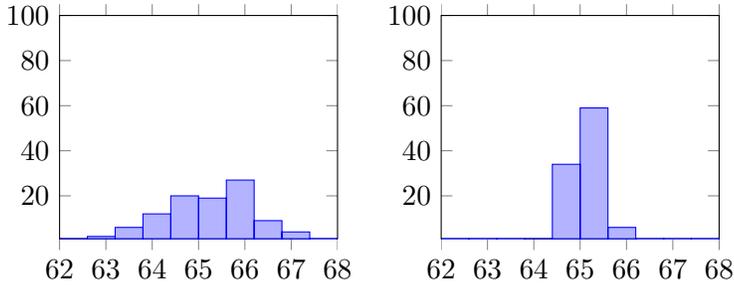
\begin{figure}[ht!]
  %\centering
  \begin{subfigure}[b]{0.2\textwidth}
    %\begin{adjustbox}{width=\linewidth} % rescale box
    \begin{tikzpicture}
\begin{axis}[
    ybar,
    ymin=1,
    ymax = 100,
    xmin = 62,
    xmax = 68,
    width=1.5\linewidth,
    xtick distance=1,
]
\addplot +[
    hist={
        %bins=5,
        %data min=62,
        %data max=68
    }   
] table[y index = 0]{datas/var_est/var_hist_ls_100000.csv};
\end{axis}
\end{tikzpicture}
\end{subfigure}
%\hfill
\hspace*{1.3cm}
\begin{subfigure}[b]{0.2\textwidth}
    \begin{tikzpicture}
\begin{axis}[
    ybar,
    ymin=1,
    ymax = 100,
    xmin = 62,
    xmax = 68,
    width=1.5\linewidth,
    xtick distance=1,
]
\addplot +[
    hist={
        %bins=2,
        %data min=62,
        %data max=68
    }   
] table[y index = 0]{datas/var_est/var_hist_ls_1000000.csv};
\end{axis}
\end{tikzpicture}
\end{subfigure}

\caption{Distribution of swing prices using Longstaff-Schwartz method. From left to right we used successively $10^{5}, 10^{6}$ simulations. Higher number of simulations leads to memory overflow.}
\label{var_hist_ls}
\end{figure}

\noindent
Figures \ref{var_hist_ep}, \ref{var_hist_nn}, and \ref{var_hist_ls} demonstrate that prices can significantly fluctuate, regardless of the pricing method. Using 1000000 simulations is insufficient to obtain a reliable price estimate. Storage limitations prevent the Longstaff-Schwartz method from exceeding this number of simulations, and this limit is even lower when the number of exercise dates is high. On the contrary, our proposed methods allow to increase the number of simulations as needed once the strategy trained. This is possible because we can evaluate our strategies on mini-batches sequentially instead of using the entire test set.

\vspace{0.2cm}
Hereafter (see Figure \ref{cpu_time}) we present CPU time for \textit{PV strat} and \textit{NN strat}. We use one batch with size of $2^{14}$ and the same PSGLD setting.

\newpage

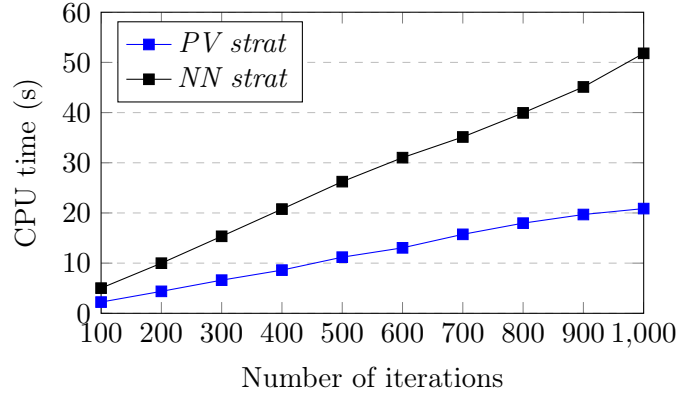
\begin{figure}[ht!]
\centering
\begin{tikzpicture}
\begin{axis}[
    xlabel={Number of iterations},
    ylabel={CPU time (s)},
    xmin=100, xmax=1000,
    ymin=0, ymax=60,
    xtick={100,200,300,400,500,600,700,800,900,1000},
    ytick={0,10,20,30,40,50,60},
    legend pos=north west,
    ymajorgrids=true,
    grid style=dashed,
   	width=0.5\linewidth,
	height=0.2\paperheight,
]

\addplot[
    color=blue,
    mark=square*,
    ]
    coordinates {
    (100,2.24)(200,4.38)(300,6.61)(400,8.61)(500,11.19)(600,13.05)(700,15.74)(800,17.97)(900,19.69)(1000,20.87)
    };
    %\legend{Strategy 1};
    \addlegendentry{\textit{PV strat}};

\addplot[
    color=black,
    mark=square*,
    ]
    coordinates {
    (100,5.02)(200,9.99)(300,15.35)(400,20.77)(500,26.24)(600,31.02)(700,35.14)(800,39.94)(900,45.11)(1000,51.81)
    };
    %\legend{Strategy 2}
    \addlegendentry{\textit{NN strat}};
    
\end{axis}

\end{tikzpicture}
\caption{CPU time (in seconds) as a function of the number of iterations.}
\label{cpu_time}
\end{figure}

\section{Summary tables: Adam \textit{versus} PSGLD}
\label{summary_table_algo}

\indent

Hereafter are recorded all the results obtained using Adam and PSGLD updatings with different hyperparameters combinations. Prices are computed by averaging 100 replications of prices; each obtained using a sample size of $M_{e} = 10^6$. We consider a swing contract with \textit{case 1} setting. Results are recorded in tables \ref{results_explicit_param}, \ref{results_nn_param_1}, \ref{results_psgld_explicit_param}, \ref{results_psgld_nn_param} where we denote by $n_b$ the number of batches used.

\begin{table}[ht]
    \centering
\begin{tblr}{colspec={c},hlines}
\hline
     $B \times n_b$ & $\gamma$ & $N$ &  $\widehat{P_0}$ &  $\widehat{P_0}^{BG}$ & Time (s) \\
     \hline
     $2^{14} \times 1$ & $0.1$ & 3000 & 65.11 ([65.04, 65.18])& 65.15 ([65.08, 65.22]) & 68.4\\
     $2^{14} \times 1$ & $0.1$ & 1000 & 65.02 ([64.95, 65.08])& 65.18 ([65.12, 65.24]) & 21.3\\
     $2^{12} \times 4$ & $0.1$ & 3000 & 65.14 ([65.08, 65.21])& 65.13 ([65.07, 65.20]) & 261.2\\
     $2^{12} \times 4$ & $0.1$ & 1000 & 65.15 ([65.08, 65.22])& 65.17 ([65.11, 65.24]) & 88.8\\
     $2^{14} \times 1$ & $0.01$ & 3000 & 64.73 ([64.63, 64.79])& 65.05 ([64.98, 65.11]) & 64.4\\
     $2^{14} \times 1$ & $0.01$ & 1000 & 63.91 ([63.83, 63.98])& 64.70 ([64.62, 64.77]) & 22.1\\
     $2^{12} \times 4$ & $0.01$ & 3000 & 65.11 ([65.04, 65.18])& 65.15 ([65.08, 65.22]) & 258.2\\
     $2^{12} \times 4$ & $0.01$ & 1000 & 64.81 ([64.74, 64.88])& 65.03 ([64.98, 65.12]) & 88.3\\
\end{tblr}
\caption{Summary table for \textit{PV strat} using Adam. The values in brackets are confidence intervals (95\%). The column \q{time} includes both training and valuation time.}
\label{results_explicit_param}
\end{table}

\newpage
\begin{table}[ht]
    \centering
\begin{tblr}{colspec={c},hlines}
\hline
     $B \times n_b$ & $\gamma$ & $N$ &  $\widehat{P_0}$ &  $\widehat{P_0}^{BG}$ & Time (s) \\
     \hline
     $2^{14} \times 1$ & $0.1$ & 3000 & 65.22 ([65.16, 65.28])& 65.23 ([65.17, 65.29]) & 150.4\\
     $2^{14} \times 1$ & $0.1$ & 1000 & 65.13 ([65.05, 65.20])& 65.22 ([65.15, 65.29]) & 53.2\\
     $2^{12} \times 4$ & $0.1$ & 3000 & 65.26 ([65.21, 65.32])& 65.23 ([65.17, 65.29]) & 603.4\\
     $2^{12} \times 4$ & $0.1$ & 1000 & 65.20 ([65.14, 65.26])& 65.20 ([65.14, 65.26]) & 133.5\\
     $2^{14} \times 1$ & $0.01$ & 3000 & 64.78 ([62.65, 64.90])& 65.06 ([64.98, 65.14]) & 157.1\\
     $2^{14} \times 1$ & $0.01$ & 1000 & 64.81 ([64.74, 64.88])& 65.05 ([64.98, 65.12]) & 55.1\\
     $2^{12} \times 4$ & $0.01$ & 3000 & 65.20 ([65.15, 65.26])& 65.20 ([65.14, 65.26]) & 607.3\\
     $2^{12} \times 4$ & $0.01$ & 1000 & 65.02 ([64.92, 65.11])& 65.17 ([65.09, 65.24]) & 135.4\\
\end{tblr}
\caption{Summary table for \textit{NN strat} using Adam. We used a neural network architecture as follows: 2 hidden layers ($I = 2$) and 10 units per layer ($q_1 = 10, q_2 = 10$). The values in brackets are confidence intervals (95\%). The time includes the training and the valuation time.}
\label{results_nn_param_1}
\end{table}

\begin{table}[ht!]
    \centering
\begin{tblr}{colspec={c},hlines}
\hline
     $\sigma$ & $\beta$  &  $\widehat{P_0}$ &  $\widehat{P_0}^{BG}$ & Time (s) \\
     \hline
      $1e^{-5}$ & 0.8  & 65.15 ([65.09, 65.22])& 65.17 ([65.10, 65.23]) & 22.6\\
     $1e^{-5}$ & 0.7  & 65.13 ([65.05, 65.21])& 65.15 ([65.07, 65.22]) & 22.1\\
      $1e^{-5}$ & 0.9 & 65.22 ([65.16, 65.28])& 65.23 ([65.18, 65.29]) & 21.6\\
     $1e^{-6}$ & 0.8  & 65.19 ([65.12, 65.26])& 65.21 ([65.13, 65.28]) & 21.3\\
     $1e^{-6}$ & 0.9  & 65.21 ([65.15, 65.28])& 65.23 ([65.17, 65.29]) & 21.4\\
\end{tblr}
\caption{Summary table for \textit{PV strat} using PSGLD. The values in brackets are confidence intervals (95\%). The column \q{time} includes both training and valuation time. We used a learning rate equal to 0.1, one batch of size $B =2^{14}$, $\lambda = 1 \cdot e^{-10}$ and 1000 iterations.}
\label{results_psgld_explicit_param}
\end{table}

\begin{table}[ht!]
    \centering
\begin{tblr}{colspec={c},hlines}
\hline
     $\sigma$ & $\beta$ &  $\widehat{P_0}$ &  $\widehat{P_0}^{BG}$ & Time (s) \\
     \hline
      $1e^{-5}$ & 0.8 & 65.18 ([65.12, 65.24])& 65.16 ([65.10, 65.22]) & 53.6\\
      $1e^{-5}$ & 0.7 & 65.22 ([65.15, 65.30])& 65.20 ([65.13, 65.27]) & 52.1\\
      $1e^{-5}$ & 0.9 & 65.23 ([65.14, 65.31])& 65.22 ([65.13, 65.30]) & 52.3\\
     $1e^{-6}$ & 0.8 & 65.27 ([65.20, 65.35])& 65.26 ([65.18, 65.33]) & 51.3\\
     $1e^{-6}$ & 0.9 & 65.23 ([65.16, 65.30])& 65.21 ([65.14, 65.28]) & 53.4\\
\end{tblr}
\caption{Summary table for \textit{NN strat} using PSGLD. We used $I = 2$ layers with $q_1 = q_2 = 10$ units. The values in brackets are confidence intervals (95\%). The column \q{time} includes both training and valuation time. We used a learning rate equal to 0.1, one batch of size $B = 2^{14}$, $\lambda = 1 \cdot e^{-10}$ and 1000 iterations.}
\label{results_psgld_nn_param}
\end{table}

\section{Langevin based optimization}
\label{langevin_optim}

This appendix aims at providing some theoretical elements on Langevin dynamics as well as describing its link to minimization problems. Assume that the potential $J$ is at least $C^1$ with Lipschitz gradient $\nabla J$ (hence sub-quadratic) and that one can represent this gradient as the expectation of a local gradient $H$ defined on $\Theta \times \mathbb{R}^p$
\begin{equation}
\label{V_as_E_loc_grad}
 \nabla J(\theta) = \mathbb{E}\big(H(\theta, Z)\big),   
\end{equation}
where $Z$ is a $\mathbb{R}^p$-valued random vector, supposed to be simulable (it can be picked up at random a datum in a dataset). Then the resulting Stochastic Gradient Descent (SGD) reads
\begin{equation}
\label{eq:SGD1}
\theta_{n+1} = \theta_n -\gamma H(\theta_n, Z_{n+1}), \; n\ge 0,\quad \theta_0 \! \in \Theta
\end{equation}
which can be rewritten as 

\begin{equation}
\label{eq:SGD2}
\theta_{n+1} = \theta_n -\gamma \nabla J(\theta_n)+ \gamma \Delta M_{n+1}, \; n\ge 0,\quad \theta_0\!\in \Theta,
\end{equation}
where $\Delta M_{n+1} = \nabla J(\theta_n)- H(\theta_n,Z_{n+1})$, $n\ge 1$,  is a sequence of martingale increments (with respect to  the filtration ${\cal F}_n = \sigma(\theta_0,Z_1,\ldots,Z_n)$, $n\ge 1$).

\vspace{0.2cm}
The idea of the Langevin dynamics is to inject an {\em independent} $d$-dimensional centered (Gaussian) white noise $(W_n)_{n\ge 1}$ by considering 
\begin{equation}
\label{eq:Langevin_0}
\theta_{n+1} = \theta_n -\gamma \nabla J(\theta_n)+ \gamma \Delta M_{n+1} +\sqrt{2\gamma}\,\sigma W_{n+1}, \; n\ge 0,\quad \theta_0\!\in \Theta.
\end{equation}
The impact of this exogenous noise  $\sqrt{\gamma} \sigma W_{n+1}$ with covariance matrix $\gamma\sigma^2I_d$ compared to that of the endogenous noise whose variance reads $ {\rm Var}( \gamma \Delta M_{n+1})= \gamma^2\|\Delta M_{n+1}\|^2$ (under standard assumptions of Robbins-Siegmund type, the martingale increments  are $L^2$-bounded) is clearly prominent when $\gamma,\sigma\to 0$ as long as $\gamma= o(\sigma^2)$ since then $\gamma^2=o(\gamma)$. One shows that (see \cite{bras:hal-03891234}) this martingale term  can be neglected, mathematically speaking, in the analysis of the procedure. So one can simply focus on the procedure
\begin{equation}
\label{eq:Langevin_0bis}
\theta_{n+1} = \theta_n -\gamma \nabla J(\theta_n)+\sqrt{2\gamma}\,\sigma W_{n+1}, \; n\ge 0,\quad \theta_0\!\in \Theta
\end{equation}
which  turns to be the Euler scheme with step $\gamma$ of the so-called Langevin dynamics
\begin{equation}
\label{langevin_sde}
d\theta_t = - \nabla J(\theta_t) dt + \sqrt{2}\sigma dB_t,\quad \theta_0 \!\in \Theta,
\end{equation}
\noindent
where $B$ is a $p$-dimensional Brownian motion. 

It is classical background that, for $\Theta = \mathbb{R}^d$, if $\int_{\mathbb{R}^d}e^{-\frac{J(\xi)}{\sigma^2}}d\xi<+\infty$,  then $\theta_t$ (in \eqref{langevin_sde}) converges in distribution (and for any ${\cal W}_p$-Wasserstein distance) toward its invariant distribution, the Gibbs measure $\nu^{(\sigma)} = C_{_{J,\sigma}}e^{-\frac{J}{\sigma^2}}\cdot\lambda_d$ (where $\lambda_d$ is the Lebesgue measure on $\mathbb{R}^d$). Similarly, if $\gamma$ is small enough, then $\theta_n$ converges in distribution toward an invariant distribution $\nu^{(\sigma)}_{\gamma}$ close to  $\nu^{(\sigma)}$. 

Note that,  if the constant step $\gamma$ is replaced in~\eqref{eq:Langevin_0bis} by a decreasing sequence $(\gamma_n)_{n\ge 1}$ such that $\Gamma_n =\gamma_1+\cdots +\gamma_n \to +\infty$ and $\sum_{n\ge 1}\frac {\gamma_n}{\Gamma_n^2}<+\infty$, then $\theta_n$ converges in distribution toward  $\nu^{(\sigma)}$ (see~\cite{BasakB1992}, see also~\cite{LambertonP2002}).

The use of such Langevin dynamics based algorithms for optimisation is justified by the fact that recent studies (see \cite{Dalalyan2014TheoreticalGF,Dalalyan2017FurtherAS}) stressed the fact  that sampling from a distribution concentrated around the global minimum or global minima of $J$ is a similar task as minimising $J$ via certain optimisation algorithms. This is theoretically justified by the fact that  $\nu^{(\sigma)}$ weakly converges to the set ${\cal P}({\rm argmin}\, J)$ of ${\rm argmin}\, J$-supported distribution as $\sigma\to 0$ and if ${\rm argmin}\, J=\{\theta^*\}$, then $\nu^{(\sigma)}\stackrel{w}{\to} \delta_{\theta^*}$. More general results, when $J$ is $C^2$ and precise results about the rate of this convergence even in the case where the Hessian of $J$ is degenerated on ${\rm argmin}\, J$ can be found in the recent paper~\cite{Bras2022}.

To still  improve the convergence and in particular to help the (SGLD) procedure escape from local minima,  practitioners introduced some  {\em preconditioners} (see~\cite{Li2015PreconditionedSG}) by  making $\sigma$ depend on  $\theta_t$ in~\eqref{langevin_sde} (usually through a function of $\nabla J(\theta_t)$ or $J(\theta_t)$ itself). A theoretical background has been provided in~\cite{Pags2020UnadjustedLA} to explain this heuristics. First, the  authors noted that  the diffusion
\begin{equation}
\label{sde_gen}
d\theta_t = b(\theta_t)dt +\sqrt{2}\sigma\, \vartheta(\theta_t) dB_t, \hspace{0.2cm} \theta_0\!\in \mathbb{R}^d
\end{equation}
%
%\noindent
%with $W$ a standard $q$-dimensional Brownian motion and a drift term of the form
where 
\begin{equation}\label{eq:newb}
b := - \Big((\vartheta \vartheta^\top) \nabla J - \sigma^2\Big[\sum_{j = 1}^{d} \partial_{\theta_j} (\vartheta \vartheta^\top)_{ij} \Big]_{i =1:d} \Big)
\end{equation} 
also has $\nu^{(\sigma)}$ as a unique invariant distribution (under an ellipticity assumption on the preconditioner $\vartheta$). The implementable  version of~\eqref{eq:Langevin_0bis}, known as PSGLD, is simply the  standard Euler scheme of this SDE with step $\gamma>0$
\begin{equation}\label{eq:Langevin_0ter}
\theta_{n+1} = \theta_n -\gamma  b (\theta_n)+\sqrt{2\gamma}\,\sigma \vartheta (\theta_n)W_{n+1}, \; n\ge 0,\quad y_0\!\in \mathbb{R}^d.
\end{equation}
 (In fact the version investigated in~\cite{Pags2020UnadjustedLA} has a decreasing step $(\gamma_n)_{n\ge 1}$  and directly converges toward $\nu^{(\sigma)}$). 
 
 Practitioners usually consider diagonal preconditioners of the form
\[
\forall\, \theta=(\theta^1,\ldots,\theta^d)\!\in \mathbb{R}^d, \quad \vartheta \vartheta^\top(\theta) = {\rm Diag}\Big(\big(\varphi(\partial_{\theta^1}J(\theta))\big)^2,\cdots, \big(\varphi(\partial_{\theta^d}J(\theta))\big)^2\Big).
\]

Numerical experiments carried out by practitioners suggest that  the resulting correcting term in the drift, which is computationally demanding in terms of complexity,  can be neglected in~\eqref{eq:newb} without damage for practical implementation (see \cite{Li2015PreconditionedSG}).

\medskip A simulated annealing version of the procedure~\eqref{eq:Langevin_0bis} has been introduced and analysed in~\cite{bras:hal-03891234} in which   $\gamma$ and $\sigma$ are  no longer constant but of the form $\gamma_n \downarrow 0$, $\Gamma_n \to+\infty$ and $\sigma_n = A/\sqrt{\log (n+1)}$. In this setting,  it is proven that $\theta_n \to {\rm argmin} \, J$ in probability.

%\nocite{*}
\bibliographystyle{plain}
\bibliography{biblio.bib}

\end{document}